\documentclass[copyright,submission,creativecommons]{eptcs}

\usepackage{breakurl}
\usepackage{graphics}
\usepackage{epsfig}
\usepackage{amsmath}
\usepackage{amssymb}
\usepackage{theorem}
\usepackage[noend]{algorithmic}
\usepackage{algorithm}

\title{Reachability in Biochemical Dynamical Systems by Quantitative Discrete Approximation}
\author{L.~Brim, J.~Fabrikov\'a, S. Dra\v{z}an, and D. \v{S}afr\'anek\thanks{The work has been supported by the Grant Agency of Czech Republic grant No. 201/09/1389 and by the Czech ministry of education intent No. MSM0021622419.}
\institute{Faculty of Informatics\\
Masaryk University\\
Botanick\'a 68a, Brno, Czech Republic}
\email{safranek@fi.muni.cz}
}

\newtheorem{theorem}{Theorem}[section]
\newtheorem{definition}{Definition}[section]
\newtheorem{example}{Example}[section]
\newtheorem{lemma}{Lemma}[section]
\newtheorem{corollary}{Corollary}[section]
\newtheorem{remark}{Remark}[section]
\newtheorem{notation}{Notation}[section]

\newenvironment{proof}[1][Proof]{\begin{trivlist}
 \item[\hskip \labelsep {\bfseries #1}]}{\end{trivlist}}

\newcommand{\qed}{$\Box$}
\newcommand{\ra}[1]{\stackrel{#1}{\rightarrow}}

\begin{document}
\maketitle

\begin{abstract}
In this paper a novel computational technique for finite discrete approximation of continuous dynamical systems suitable for a significant class of biochemical dynamical systems is introduced. The method is parameterized in order to affect the imposed level of approximation provided that with increasing parameter value the approximation converges to the original continuous system. By employing this approximation technique, we present algorithms solving the reachability problem for biochemical dynamical systems. The presented method and algorithms are evaluated on several exemplary biological models and on a real case study. This is a full version of the paper published in the proceedings of CompMod 2011.
\end{abstract}

\section{Introduction}
\label{sec:introduction}

Under the modern holistic paradigm provided by \emph{systems biology}~\cite{FSBKitano}, genome-scale knowledge of individual components is combined with knowledge of interactions underlying the physiology of living organisms. The central goal of systems biology is to integrate all available biological data and to reconstruct \emph{executable models}~\cite{FH07} which allow to investigate complicated behaviour emerging from the underlying biochemistry. 
An important dimension is the quantitative aspect of the data and processes being modeled. 

With respect to~\cite{Fe87}, we consider biological models to be captured by the notion of a \emph{biochemical dynamical system} consisting of variables describing a certain quantity of the respective species in time (e.g., number of molecules or molar concentration). Variable values evolve in time with respect to rules modeling the effect of reactions. The space of all possible configurations of variable values is referred as the \emph{state space}.

There exist several modeling approaches that differ in abstraction employed for modeling of time, variable values, and molecular interaction effects. 
The most commonly used approach concerns systems of ordinary differential equations (ODE)~\cite{OP74} where both time and model variables are interpreted as continuous quantities. Effects of interactions are modeled in terms of continuous deterministic updates of variables. Variable values represent molar concentrations of the species. In general, the ODE approach relies on many physical and chemical assumptions simplifying thermodynamic conditions under which particular biochemical phenomena can be modeled correctly~\cite{K70}. It is important to note that even simple interactions such as second order reactions lead to non-linear ODEs. However, under certain assumptions, biological systems make specific subclasses of general non-linear dynamical systems. Such a specialization motivated development of specific analysis techniques~\cite{Fe87,TNO96,BBW08,KB08}.

Nevertheless, dimensionality and complexity of biological models preclude satisfactory application of analysis methods implying that to explore the model dynamics the only practicable method is numerical simulation. Since numerical simulation generates an approximate solution (a trajectory) starting from a single initial point in the continuous state space, the scope of such exploration is limited to the particular trajectory only. This is sufficient for ``local'' analysis provided that initial conditions are precisely known. However, studied systems are typically under-determined in terms of uncertain quantitative parameters and initial conditions. Therefore generalization of the exploration scope is necessary to reveal and understand the complicated emergent behaviour. An important example of a problem which cannot be effectively solved by local methods is \emph{global temporal property} -- the problem to decide whether a given dynamical phenomenon, e.g., oscillation or variables correlation, is globally present/absent for all considered initial conditions~\cite{CES86,MRM+08}. 

In this paper we limit ourselves to a subclass of dynamical phenomena representing \emph{reachability} of a given portion of the state space. Example of a global temporal property problem that belongs to this subclass is to identify minimal or maximal concentration of species reachable from a particular set of initial conditions. 

In general, the reachability problem is undecidable due to unboundedness and uncountability of the state space. However, since concentrations of species cannot expand infinitely, state spaces of biological systems dynamics can be considered bounded in most cases. Analysis can be therefore considered indirectly on suitable finite discrete approximations of continuous state spaces~\cite{KB08,BRJ+08}. 


For a significant class of biochemical dynamical systems determined by multi-affine vector fields (i.e., affine in each variable), there has been developed an over-approximative abstraction technique based on partitioning the continuous state space by a finite set of \emph{rectangles}. Rectangles determine states of a \emph{rectangular transition system} representing the finite discrete (over)approximation of the continuous state space~\cite{BH06}, as shown in Figure~\ref{fig:abstscheme}a. The rectangular abstraction has been employed in~\cite{KB08} for reachability analysis and further elaborated by model checking methods in~\cite{HIBI09}. The results show that the extent of spurious behaviour introduced by the abstraction is typically very high thus limiting satisfactory application of the method. The problem is based mainly on the fact that a transition between any two individual rectangles over-approximates the vector field on the border between the rectangles (a so-called \emph{facet}, see Figure~\ref{fig:abstscheme}b) provided that the information regarding which trajectories starting in an entry facet evolve through a particular exit facet is abstracted out. This causes the rectangular transition system to generate many rectangle sequences in which there is no corresponding trajectory of the original continuous system embedded. Moreover, the extent of such spurious behaviour is not directly eliminated by increasing the partition density.

When analysing approximate models as in systems biology, the need for precise results critically required in systems verification can be relaxed provided that a suitable approximation of the solution can be even more efficient to obtain useful results. Henceforth, in the field of complex systems exhaustive techniques are often combined with approximative methods thus making a certain shift in the way of applying formal methods~\cite{RBF+09,GP06,JCL+09}.

\subsection{Our Contribution}

We present a new technique for discrete approximation of biochemical systems with dynamics given by a system of ODEs with multi-affine right hand side.
Our discrete approximation is not an exact abstraction wrt the original continuous system, but rather an approximation that approaches exact reachability with decreasing approximation granularity. 
While still assuming the rectangular partition at the background, we employ a \emph{
measure} that enables local quantification of the amount of trajectories evolving on a rectangle in a particular facet-to-facet direction. To this end, every rectangle is augmented with a local memory representing the information at which part (\emph{entry set}) of the entry facet it has been entered. On each entry set, we identify \emph{focal subsets} from which all trajectories lead to the same exit facet. 
In Figure~\ref{fig:abstscheme}c, there are two different states of a \emph{quantitative discrete approximation automaton} (QDAA) depicted. Both states share the same rectangle $[1,1.5]\times[1,1.5]$ and they differ in entry sets (marked yellow). The upper state with entry set $\{1.5\}\times[1,1.5]$ has only one focal set - all trajectories from its entry set exit the state through the facet $[1,1.5]\times\{1\}$. 
The second state with entry set $[1,1.5]\times\{1.5\}$ has two focal subsets made by the green and the red part of the entry facet, respectively. 

Transitions from a state with given entry set have weights assigned to themselves.
Consider a transition from a state $A$ to a state $B$.
The transition exists if there is a part $P$ of the entry set of $A$ such that the trajectories of ODE solutions go from $P$ to $B$.
Weight of a transition from $A$ to $B$ corresponds to the ($n-1$-dimensional) volume of $P$ divided by the volume of the entry set of $A$.
In this manner, the measure reflects amounts of trajectories proceeding in a particular direction. Rectangle regions related by weighted transitions make the QDAA which is 
a discrete-time Markov chain. (See Theorem~\ref{thm:mc} and its proof.)

From a computational viewpoint, the continuous volumes are finitely approximated by discretization on a uniform grid. Local numerical simulations are employed to identify the entry regions and focal subsets. The density of facet discretization grid is considered as the \emph{method parameter}. Because of combining numerical simulation with rectangular abstraction, the resulting QDAA makes neither an over- nor an under-approximation of the original continuous system. Since for every sequence of states the approximate volume measure converges to the continuous volume with increasing discretization parameter, the parameter indirectly affects the correspondence between the original continuous behaviour and its approximation. This makes the method sufficient for approximating reachability in complex biochemical dynamical systems.


In general, the following main contributions are brought by this paper.

\begin{enumerate}
\item A novel computational technique for finite discrete approximation of multi-affine dynamical systems by means of QDAA.

\item Showing that QDAA converges to the original continuous system behaviour. (See Theorem~\ref{thm:kappa_lebesgue} and its proof.)

\item A reachability algorithm for QDAA.

\item Evaluation on elementary models and an \emph{E. Coli} case study.
\end{enumerate}

Since the most common application of the considered systems class is the domain of biochemical dynamical systems modeled directly by rules of mass action kinetics~\cite{HJ72}, evaluation of the method and algorithms is realized on biological models fitting this framework.

\begin{figure*}
\begin{center}
\vspace{-2mm}
\includegraphics[scale=.3]{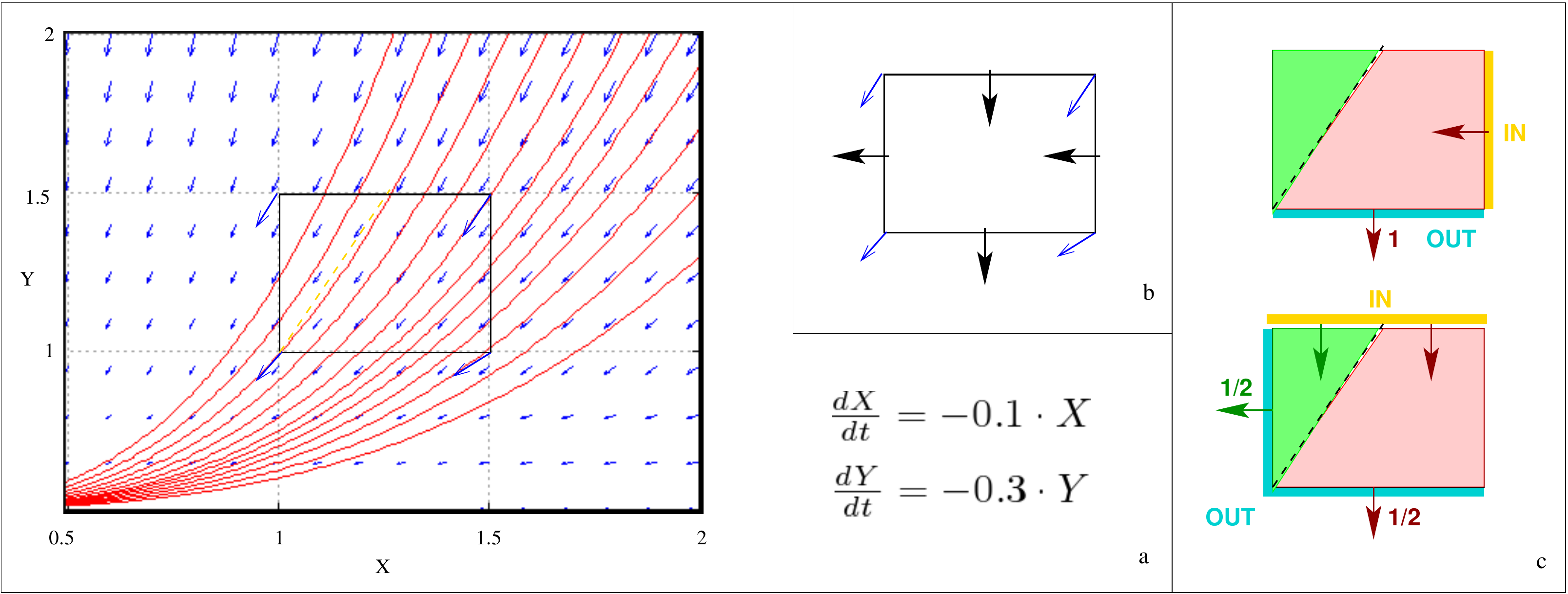}
\end{center}
\caption{(a) Vector field of a linear system partitioned by thresholds, (b) the principle of rectangular abstraction, (c) and quantification of the extent of over-approximation in terms of transition weights. The dashed line inside the rectangle demonstrates the approximate border separating trajectories exiting through different facets.}
\label{fig:abstscheme}
\end{figure*}

\subsection{Related Work}

Discrete approximation methods are commonly used in continuous and hybrid systems analysis (see~\cite{ADF+06} for an overview regarding reachability) to handle the uncountability of the state space. Direct methods work on the original system and rely on a successor operation iteratively computing the reachable set whereas indirect methods abstract from the continuous model by a finite structure for which the analysis is simpler. Our method belongs to the latter class, since 
it uses numerical simulations and creates the abstraction automaton. Considering a fixed set of initial conditions, there is a certain overhead with generating states of the automaton in comparison with simple numerical simulations. 
However, the advantage of constructing the automaton is obtaining a global view of the dynamics. Moreover, in addition to rectangular abstraction, the automaton is augmented with weighted transitions which represent quantitative information describing volumes of subsets of initial conditions belonging to attraction basins of different parts of the phase space.

An indirect method based on rectangular abstraction automaton making the finite quotient of the continuous state space has been employed, e.g., in~\cite{KB08,HKI+07,BRJ+08}. In general, these methods rely on results~\cite{BH06,HS04} and are applicable to (piece-wise) affine or (piece-wise) multi-affine systems.
Although not addressed formally in this paper, our technique can be considered as a refinement of~\cite{KB08}. However, we focus on obtaining satisfactory approximate results eliminating the extent of spurious behaviour coming from conservativeness of rectangular abstraction. 
Our technique can be employed for the recognition of spurious behaviour of the rectangular abstraction transition system. 

The technique presented in~\cite{MB08} employes timed automata for the finite quotient of a continuous system as an alternative to piece-wise linear approximations. Another indirect technique adapted to multi-affine biological models is~\cite{BHK07}. The approach also employes rectangular abstraction, but results in less conservative reachable sets by means of polyhedral operations. In~\cite{ADF+06,DHR05} there are techniques proposed for rectangular refinement that go towards reduction of over-conservativeness. These techniques work fine for linear systems while leaving the non-linear systems as a challenge.

Direct methods are mostly based on hybridization realized by partitioning the system state space into domains where the local continuous behaviour is linearized~\cite{ADG07}. This method, in an improved form, has been applied to non-linear biochemical dynamical systems~\cite{DLM+09}. In general, direct methods give good results for low-dimensional systems and small initial sets. In comparison with indirect approaches, they are computationally harder. From this viewpoint, our approach lies between both extremes.



\section{Preliminaries}
\label{sec:preliminaries}

\subsection{Basic definitions and facts}

Let $\mathbb{N}$ denote the set of positive integers, $\mathbb{N}_0$ the set $\mathbb{N}\cup\{0\}$, and $\mathbb{R}^+_0$ the set of nonnegative real numbers. For $n \in \mathbb{N}$, denote $\mathbb{R}^n$ the standard $n$-dimensional Euclidean space with standard topology and Euclidean norm $\left|\cdot\right|: \mathbb{R}^n \rightarrow \mathbb{R}^+_0$. For an arbitrary function $f$ we use the common notation $\mathit{dom}(f)$ for the domain of $f$.

For every $i\in\{1,\dotsc, n\}$ assume $a_i,b_i \in \mathbb{R}$ such that $a_i \leq b_i$. Denote $I = \prod_{i=1}^n [a_i,b_i]$ an \emph{$n$-dimensional closed interval in} $\mathbb{R}^n$ and $\mathit{vol}(I)$ the \emph{$n$-dimensional volume of $I$} defined as $\mathit{vol}(I)= \prod_{i=1}^n (b_i - a_i)$. Further denote $\mathit{Inter}(I)$ the \emph{interior of $I$}, defined as the cartesian product of open intervals $\prod_{i=1}^n (a_i,b_i)$.

For any $X \subseteq \mathbb{R}^n$ denote $\lambda^*_n(X)$ the \emph{Lebesgue outer measure} (on $\mathbb{R}^n$) of the set $X$. Basically $\lambda^*_n(X)$ is the minimal nonnegative real number such that whenever $X$ can be covered by a sequence of closed intervals in $\mathbb{R}^n$ the sum of volumes of these intervals is greater then or equal to $\lambda^*_n(X)$.
(For precise definitions see~\cite{WalterRudinMA}.)
Note that $\lambda^*_n(X) < \infty$ for every bounded set $X$ and $\lambda^*_n(I)=\mathit{vol}(I)$ for every $n$-dimensional interval $I$.

Let $n\geq 2, i \leq n, c \in \mathbb{R}$. We use
$\mathbb{R}^{n-1}_i(c)$ to denote the hyper-plane $\mathbb{R}^{n-1}_i(c)=\{\langle x_1,\dotsc,x_n \rangle \in \mathbb{R}^n \mid x_i = c\}$. Denote $\hat{\pi}_i:\mathbb{R}^n\rightarrow \mathbb{R}^{n-1}$ the projection omitting the $i$th variable,
$\hat{\pi}_i(\langle x_1, \dotsc, x_n \rangle) = \langle x_1, \dotsc, x_{i-1}, x_{i+1}, \dotsc x_n\rangle$. Let $X \subseteq \mathbb{R}^{n-1}_i(c)$. We extend the notion of the $(n-1)$-dimensional Lebesgue outer measure to such sets $X$ and denote $\lambda^*_{n-1}(X)$ the $(n-1)$-dimensional Lebesgue outer measure of $\hat{\pi}_i(X)$.
 

Let $f: \mathbb{R}^{n} \rightarrow \mathbb{R}^{n}$ be a continuous function (an autonomous vector field). We say that 
\begin{equation}
 \label{eq:autsystem}
\dot{x} = f(x)
\end{equation}
is an \emph{autonomous ODE system}. 
An important property of autonomous systems is the fact that if $y(t)$ is a solution of~(\ref{eq:autsystem}) on an open interval $(a,b)$, then $y(t+t_0)$ is also a solution (defined on interval $(a-t_0,b-t_0)$).

 A function $f:\mathbb{R}^n \rightarrow \mathbb{R}^n$ \emph{satisfies the Lipschitz condition locally} on $\mathbb{R}^n$, if for every $x\in \mathbb{R}^n$ there exists an open set $U \subseteq \mathbb{R}^n$, $x\in U$ and a constant $L \in \mathbb{R}$ such that for every two points $x_1, x_2 \in U$ the inequality 
$\left|f(x_1) - f(x_2) \right| \leq L \cdot \left| x_1 - x_2 \right|$ holds.

\begin{theorem}[Trajectories of solutions of an autonomous system]
 \label{thm:autonomous_traj}
Let~(\ref{eq:autsystem})
be an autonomous ODE system, where $f$ is defined on $\mathbb{R}^n$ and let $f$ satisfy the Lipschitz condition locally on $\mathbb{R}^n$. Let $x$ be an inextendible solution of system~(\ref{eq:autsystem}). Then $dom(x)$ is an open interval, 
and for every point $\alpha\in\mathbb{R}^n$ there exists exactly one trajectory of an inextendible solution $x(t)$ of system~(\ref{eq:autsystem}) coming through $\alpha$.
%
\end{theorem}

\begin{theorem} [Continuous dependency on initial conditions]
\label{thm:continuous_dependence}
Let $f:\mathbb{R}^n \rightarrow \mathbb{R}^n$ be continuous on an open set $E \subseteq \mathbb{R}^n$
with the property that for every $y_0\in E$, the initial value problem
$\dot{x} = f(x), x(0)=y_0$
has a unique solution $y(t)= \eta(t,y_0)$ ($\eta$ is a function of variables $t,y_0$ ). 
Let $w_{\bot}, w_{\top}\in \mathbb{R}$ such that $(w_{\bot},w_{\top})$ is 
the maximal interval of existence of $y(t)=\eta(t,y_0)$.

Then the bounds $w_{\bot}, w_{\top}$  are (lower, resp. upper semicontinuous) functions of $y_0$ in $E$ and 
$\eta(t,y_0)$ is continuous on the set
$\{\langle t, y_0\rangle\mid y_0 \in E, w_{\bot}(y_0) < t < w_{\top}(y_0) \} \subseteq \mathbb{R}^{n+1}$.
\end{theorem}
%

We restrict ourselves to multi-affine autonomous systems. That is, systems of the form~(\ref{eq:autsystem}), such that the vector field $f$ is a \emph{multi-affine} function, defined as a polynomial of variables $x_1,\dotsc,x_n\in\mathbb{R}^n$ of degree at most one in every variable.
 The assumptions of  Theorems~\ref{thm:autonomous_traj} and \ref{thm:continuous_dependence} 
(from~\cite{HartmanODE}) are satisfied for systems of this class, therefore the properties stated in the above theorems can be used for reasoning about autonomous systems with multi-affine vector fields.





\subsection{Biochemical dynamical system}
According to~\cite{Fe87}, by a biochemical dynamical system we understand a collection of $n$ biochemical species interacting in biochemical reactions. Species concentrations are represented by variables $x_1,\dotsc, x_n$ attaining values from $\mathbb{R}_0^+$. If the stoichiometric coefficients in reactions do not exceed one and the reaction dynamics respects the law of mass action kinetics~\cite{HJ72}, the dynamical system can be described by a multi-affine autonomous system in the form~(\ref{eq:autsystem}).

In a biochemical dynamical system we are typically interested in a bounded part ($n$-dimensional interval) of the phase space in $\mathbb{R}^n$.
Further, we consider the phase space partitioned by a (non-uniform) rectangular grid. In particular, for each variable there is defined a finite set of \emph{thresholds}, making the system \emph{partition}. Thresholds determine $(n-1)$-dimensional hyper-planes in $\mathbb{R}^{n}$ and can be freely specified according to particular questions that should be addressed by the model analysis, e.g., specification of unsafe or attracting sets. Cells laid out by $2n$ adjacent threshold hyper-planes (cells are again intervals in $\mathbb{R}^n$) are called \emph{hyper-rectangles}, for short we refer to them as \emph{rectangles}.

\begin{definition}Define a \emph{biochemical dynamical system} (\emph{biochemical system} for short) as a tuple
$\mathcal{B} = \langle n,f,\mathcal{T},\mathcal{I}_C\rangle$, where
\begin{itemize}
\item $n\in \mathbb{N}$ is the \emph{dimension} of $\mathcal{B}$,
\item $f:\mathbb{R}^n\rightarrow \mathbb{R}^n$ is the multi-affine vector field of $\mathcal{B}$, 
\item $\mathcal{T}=\langle T_1,\dotsc,T_n \rangle$ is the \emph{partition} of $\mathcal{B}$ where each $T_i$ is a finite subset of $\mathbb{R}^+_0$, and define the \emph{set of rectangles} given by $\mathcal{T}$ as 
$$\mathit{Rect}(\mathcal{T})=\{ \prod_{j=1}^n I_j\mid \forall j \exists a,b \in T_j : I_j=[a,b], \forall c \in T_j: c \leq a \vee c \geq b \},$$
\item 
$\mathcal{I}_C \subseteq \mathit{Rect}(\mathcal{T})$ is the  set of \emph{initial conditions} (\emph{initial set}) of $\mathcal{B}$.
\end{itemize}
\end{definition}

%

\begin{definition}
Let $\mathcal{B}=\langle n,f,\mathcal{T}, \mathcal{I}_C\rangle$ be a biochemical system and let $H \in Rect(\mathcal{T})$ be a rectangle such that $H=I_1\times \ldots \times I_n$, where
$I_i = [a_i, b_i]$. 
For every $i \in \{1,\ldots, n\} $ define
the \emph{lower} (resp. \emph{upper}) \emph{facet} of $H$ wrt the $i$th variable:
$$\begin{array}{l}
\mathit{Facet}^{\bot}_i(H) = \{ \langle x_1, \dotsc, x_n \rangle \in H \mid x_i = a_i\},\\ 
\mathit{Facet}^{\top}_i(H) = \{ \langle x_1, \dotsc, x_n \rangle \in H \mid x_i = b_i\}.
\end{array}$$

Denote $Facets_i(H)$ the \emph{set of $i$th dimension facets of H}, $Facets_i(H)=Facet^{\bot}_i(H)\cup Facet^{\top}_i(H)$, and 
$Facets(H)$
the \emph{set of (all) facets of $H$}, $Facets(H) = \bigcup_{i=1}^n Facets_i(H)$.
\end{definition}


\begin{definition}
Let $H$, $H'\in\mathit{Rect}(\mathcal{T})$. We say that $H$ is a \emph{neighbour of} $H'$, denoted $H\bowtie H'$,
if there exists $F \in \mathrm{Facets}(H)$ such that $H \cap H' = F$. 
\end{definition}

\section{Quantitative Discrete Approximation}
\label{sec:abstraction}

\begin{figure}[!h]
\raisebox{1.8cm}{\scalebox{.8}{\parbox{1cm}{
$$\begin{array}{c}
A \stackrel{0.5}{\rightarrow} B\\
B \stackrel{0.8}{\rightarrow} A\\
\end{array}
$$
}}}
\hspace{0.5cm}
\raisebox{1.8cm}{\scalebox{.8}{\parbox{4.5cm}{%
$$\begin{array}{c}
\frac{d[A]}{dt}=-0.5\cdot [A] + 0.8\cdot [B]\\[1mm]
\frac{d[B]}{dt}=0.5\cdot [A] - 0.8\cdot [B]\\[4mm]
\text{thresholds on } [A]: \{0,2.5,5\}\\
\text{thresholds on } [B]: \{0,2.5,5\}
\end{array}
$$
}}}
\hspace{1.5cm}
\includegraphics[scale=.25]{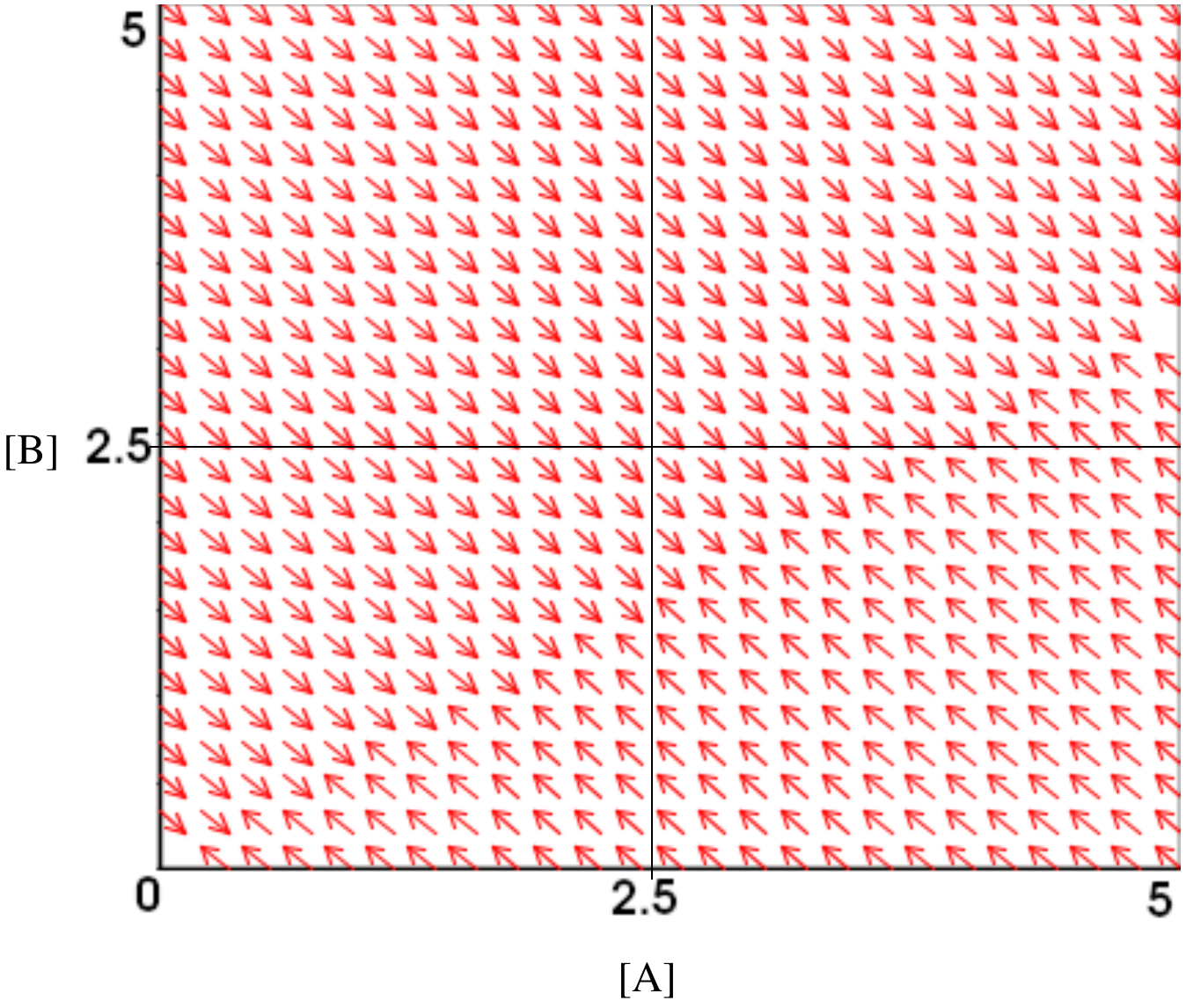}
\hspace{1.5cm}
\includegraphics[scale=.25]{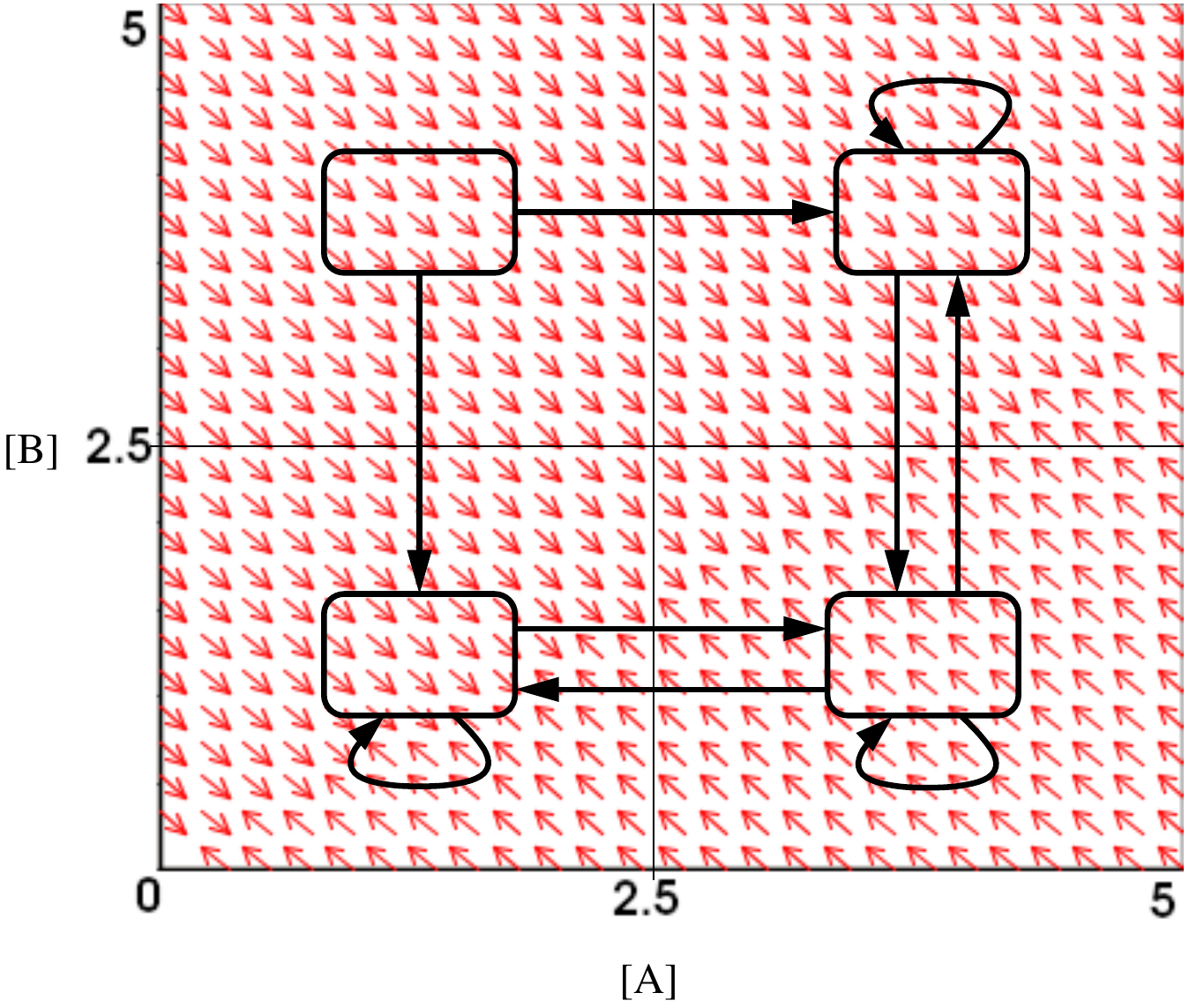}
\caption{Example of a biochemical system with two species and two reactions. Dynamics given by a system of two ODEs and the system of thresholds are in the left part of the figure.
Vector field is visualized in the middle, and its Rectangular Abstraction Transition System on the right.}
\label{fig:example}
\end{figure}

Given a biochemical system $\mathcal{B}=\langle n,f,\mathcal{T},\mathcal{I}\rangle $, we aim to define a finite automaton reflecting the behaviour of $\mathcal{B}$, and for each state, to assign every transition a weight quantifying probability of proceeding to a particular successor.




A state is defined as a pair $\langle H,E\rangle$ -- a rectangle $H$, and a subset $E$ of a particular facet of $H$. The set $E$ represents a so-called \emph{entry set}, a region through which trajectories of the system~(\ref{eq:autsystem}) enter the interior of $H$. Intuitively, we can say that $E$ encodes the history of previous evolution of the system from initial set $\mathcal{I}_C$ to $H$.  
Entry sets are either subsets of $(n-1)$-dimensional facets of $H$ or (in case of initial states) the whole $n$-dimensional rectangle $H$.

Since entry sets can be arbitrary sets in Euclidean space, we approximate them by a finite discrete structure. Each facet is provided with a uniform grid on which we approximate any subset of the facet by the set of rectangular fragments, so-called \emph{tiles} (Figure~\ref{fig:tiles}). The grid is $n$-dimensional or $(n-1)$ dimensional depending on the dimension of approximated entry sets. When following the trajectories of solutions of differential equations of the models dynamics in time, entry sets are identified by trajectories of solutions passing through them on their way from preceding rectangles. In following definitions we treat this intuitive perception of entry sets formally.

Let $\kappa \in \mathbb{N}$, let $\mathcal{B}=\langle n,f,\mathcal{T},\mathcal{I}_C\rangle$ be a biochemical system, $H\in \mathit{Rect}(\mathcal{T})$, and $F \in \mathit{Facets}(H)$ for all definitions and theorems from this section.

\begin{definition}
Let $H$ be of the form $H=\prod_{j=1}^n I_j$, where $\forall j: I_j=[a_j,b_j]$.
Let $B \in \{ H\} \cup \mathit{Facets}(H)$. Set either $n'=n$, if $B=H$, or $n'=(n-1)$, if $B \in \mathit{Facets}_i(H)$ for some $1 \leq i \leq n$ (in this case $\exists c \in \{a_i,b_i\}: B \subset \mathbb{R}^{n-1}_i(c)$).

Define the \emph{set of $\kappa$-tiles} of $B$ as 
$\mathit{Tiles}^{\kappa}_{n'}(B)=\{ A \subseteq B \mid A=\prod_{j=1}^n A_j\}$, where 
$A_i = \{c\}$, if $B \in \mathit{Facets}_i(H)$, and otherwise ($j \neq i$ or $B = H$) $A_j$ is a closed interval in $\mathbb{R}^+_0$ of the form 
$[a_j + \frac{k_j}{\kappa}(b_j-a_j), a_j + \frac{k_j + 1}{\kappa} (b_j - a_j) ] $, where for all
 $j \in \{1,\dotsc,n\},  j \neq i$ the nonnegative integer $k_j \in \mathbb{N}_0$ satisfies  $k_j < \kappa$. 
\end{definition}


The following definition introduces the notion of general entry sets.

\begin{definition}
\label{def:entry}
Define the \emph{set of entry points into a rectangle $H$ through facet $F$}, 
as the set 
$\mathit{Entry}(F,H) =$\newline
\noindent
$\bigl\{  y_0 \in F \mid\exists \mathrm{\ a\ trajectory\ }y(t) \mathrm{\ of\ a\ solution\ of~(\ref{eq:autsystem})}
\mathrm{\ such\ that\ }y(0) = y_0\mathrm{\ and\ }\exists \epsilon >0 : y(t)\in H\mathrm{\ for\ }\forall t \in (0,\epsilon) \bigr\}.$
\end{definition}

Next we define the approximation of entry sets on a grid of $\kappa$-tiles. Additionally, we define the respective (discrete) volume measure of a set (see Figure~\ref{fig:tiles}~c),d)).

\begin{definition}
Let $X \subset H$. Let $n'=n-1$, if there exists $i \in \{1,\dotsc, n\}, F \in Facets_i(H)$ such that $X \subseteq F$, and let $n'=n$, otherwise. Let $M=F$, if $X \subseteq F$, and let $M=H$, if there is no such facet $F$.
Define
\itemize
\item the \emph{set of $\kappa$-tiles} approximating the set $X$ as 
$$\mathit{Tiles}^{\kappa}_{n'}(X) = \Biggl\{ A \in \mathit{Tiles}^{\kappa}_{n'}(M) \mid 
\frac{\lambda^*_{n'}(A \cap X)}{\lambda^*_{n'} (A) } \geq \frac{1}{2}\Biggr\},$$
\item the \emph{rectangular $\kappa$-grid measure} of the set $X$ as $\lambda_{n'}^{\kappa}(X)= \sum_{A \in \mathit{Tiles}^{\kappa}_{n'}(X)} \mathit{vol}(A).$
\end{definition}

\begin{figure}
\vspace{-3mm}
\includegraphics[scale=.18]{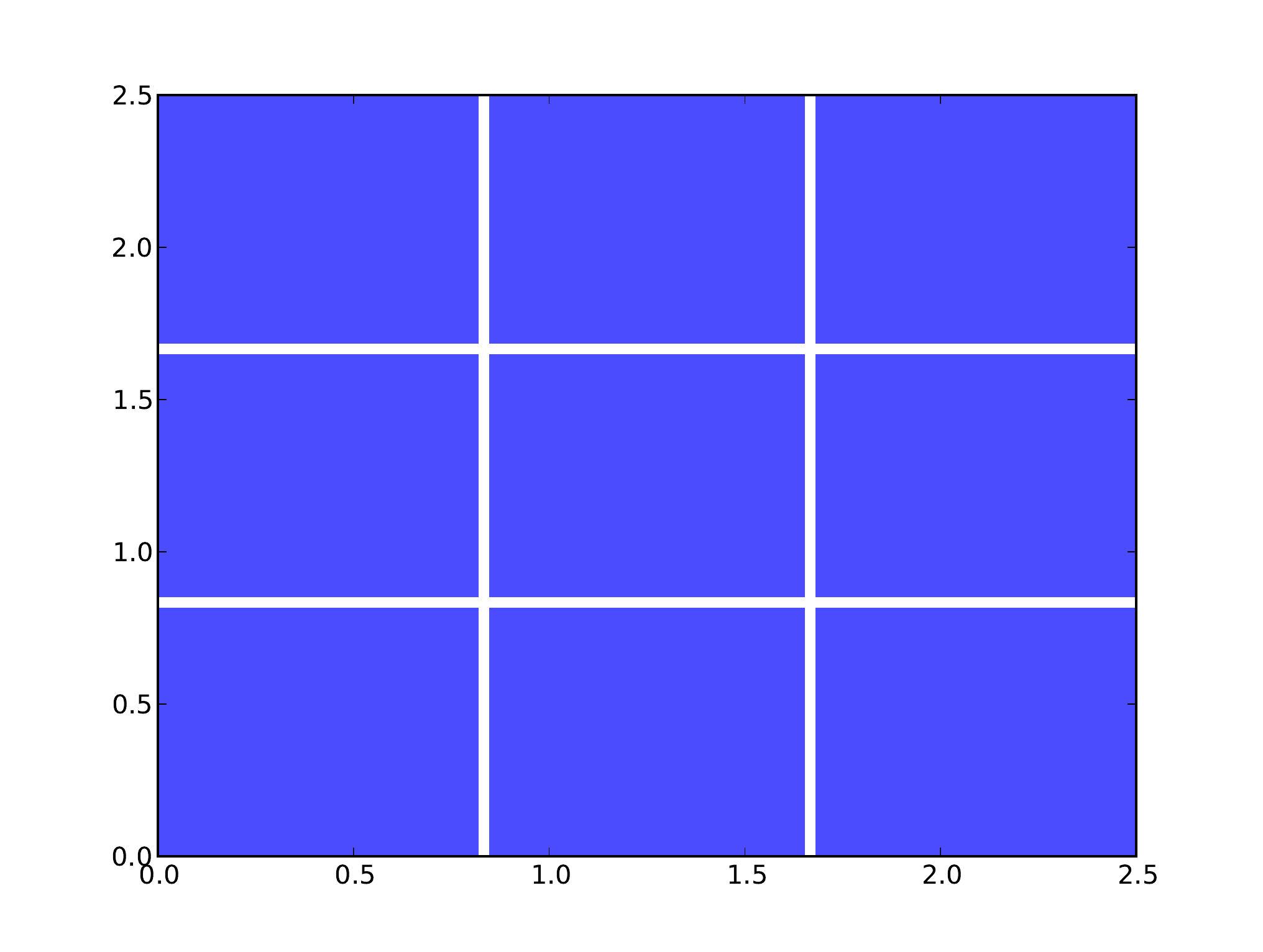}
\includegraphics[scale=.18]{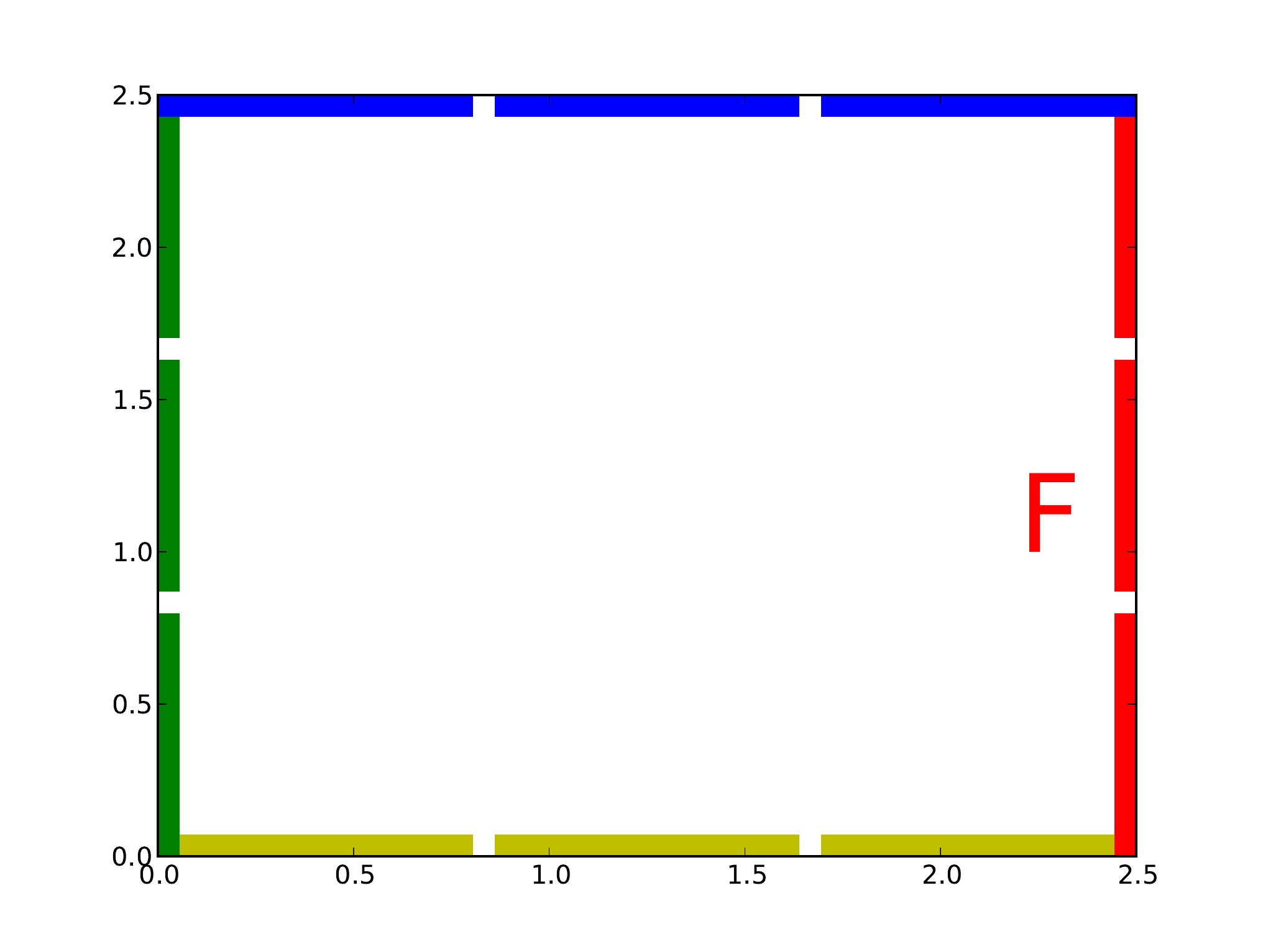}
\includegraphics[scale=.18]{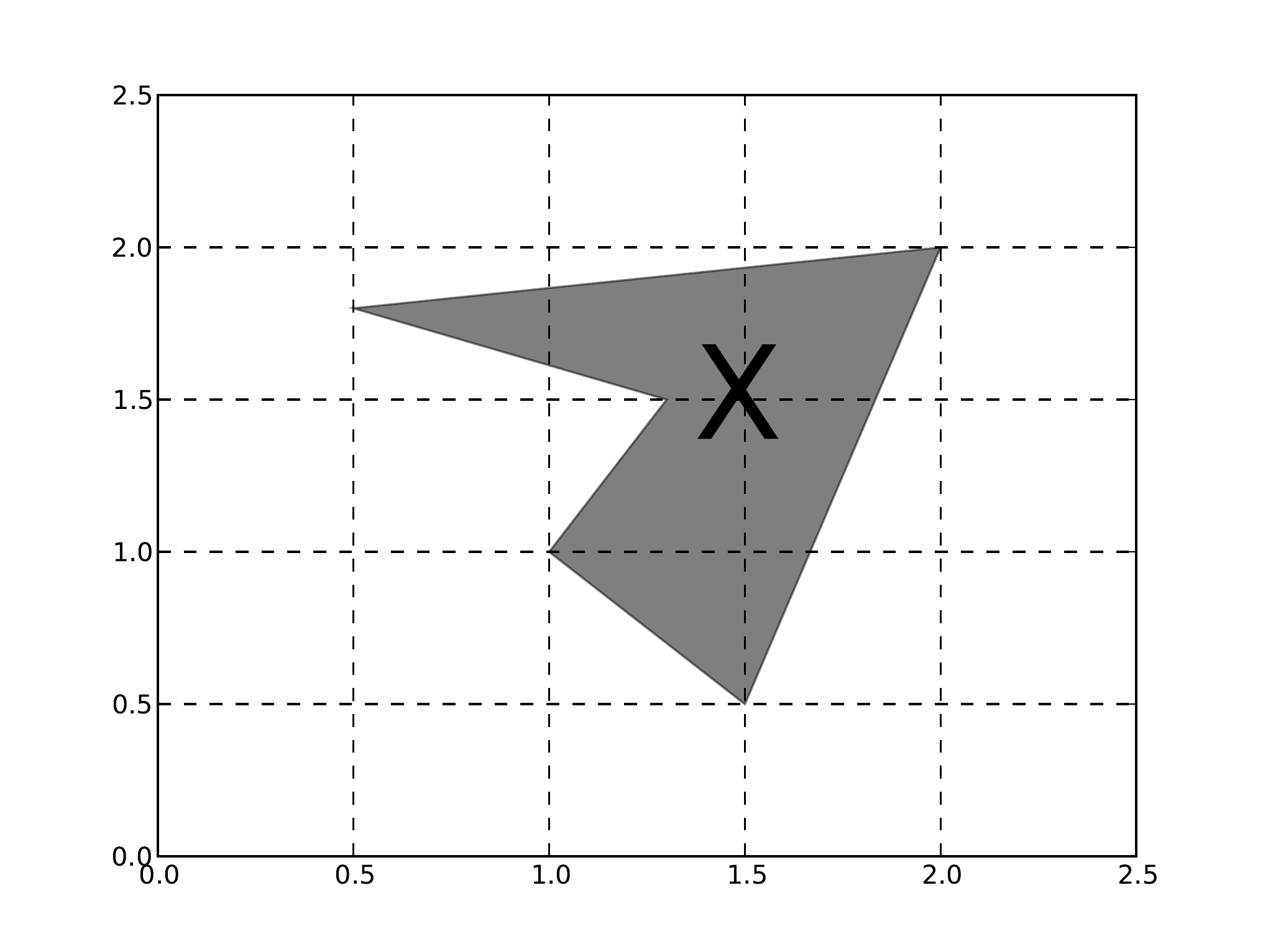}
\includegraphics[scale=.18]{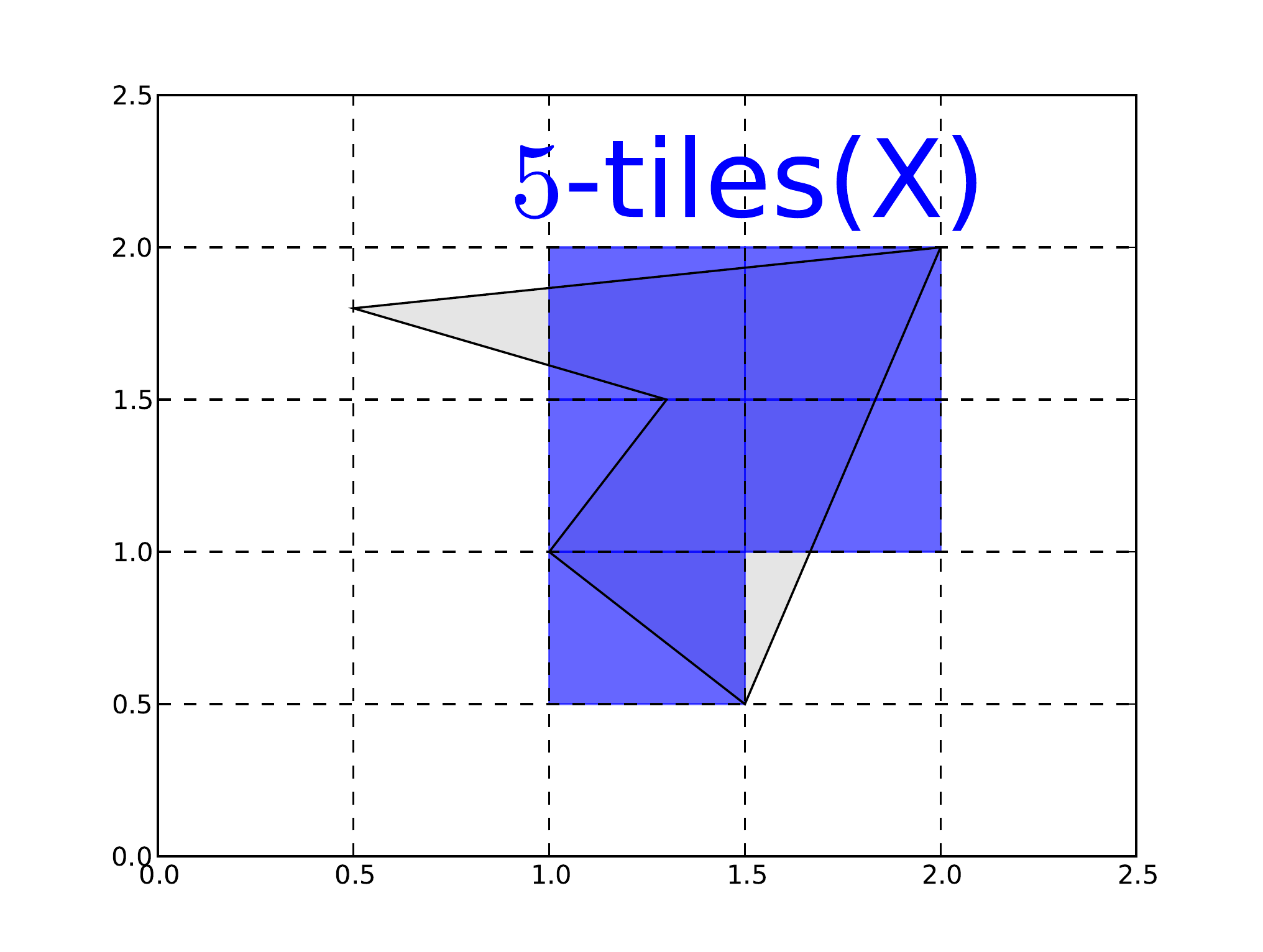}
\vspace{-4mm}
\hspace{1.5cm}
a)
\hspace{3.5cm}
b)
\hspace{3.5cm}
c)
\hspace{3.5cm}
d)
\caption{
a)
Let $H=[0,2.5]\times[0,2.5]$ be a rectangle.
The blue areas depict elements of $\mathit{Tiles}^3_2(H)$.
\newline 
b) 
Let $F=\mathit{Facet}^{\top}_1(H) = \{2.5\}\times[0,2.5]$.
The red line segments are elements of $\mathit{Tiles}^3_1(F)$.
\newline
The set $\mathit{EntrySets}_3(H)$ has $2+4\cdot(1+\binom{3}{2}+\binom{3}{1})=30$ elements:
$\emptyset, H,$ and $7$ for every facet of $H$ (the~facet itself, $3$ segments and $3$ unions of pairs of segments of the facet).
\newline
c) 
Let $X$ be a subset of $H$ (the shaded polygon). Let $\kappa=5$.
\newline
d) 
The set of $\kappa-$tiles approximating $X$ is the set of five blue intervals (each satisfying that at least half of its area is in $X$).
The cardinality of $\mathit{Tiles}^{\kappa}_2(X)$ is $5$. Thus $\lambda^{\kappa}_2(X)=5 \cdot (0.5 \cdot 0.5) = 1.25.$
}
\label{fig:tiles}
\end{figure}


The following definition declares the set of all discretized entry sets for a given rectangle. 

\begin{definition}
\label{def:entrysets}
For $H$, define \emph{set of (approximate) entry sets} $\mathit{EntrySets}_{\kappa}(H) = $\\[-2mm]
$$\Bigl\{ E \subseteq H \mid E=\emptyset \vee E=H \vee \exists F \in \mathit{Facets}(H), \mathcal{E} \subseteq \mathit{Tiles}_{\kappa}\bigl(\mathit{Entry}(F,H)\bigr) : E = \bigcup \mathcal{E} \Bigr\}.$$
\end{definition}


For an example of a set of (approximate) entry sets of a rectangle see Figure~\ref{fig:tiles}~a),b).
Note that set of approximate entry sets is always finite. Further note that also the empty set and the entire rectangle are considered as entry sets. These represent singular cases needed in the subsequent construction of the automaton. In particular, states with the empty entry set approximate fixed point behaviour not leaving the rectangle (steady state memory) whereas the rectangle-form entry set is employed for initial rectangles.

\begin{definition}
 \label{def:startsets}
Let $E \in \mathit{EntrySets}_{\kappa}(H), H' \in \mathit{Rect}(\mathcal{T}), F' \in \mathit{Facets}(H)$ such that $H' \bowtie H, F'=H\cap H'$.

Define the \emph{focal subset of $E$ on $H$ targeting $F'$}, denoted $\mathit{Focal}(H,E,F')$, as the set of all $y_0 \in E$ such that there exist $\epsilon,\epsilon',c >0$
and a trajectory of a solution $y(t)$ of system~(\ref{eq:autsystem}) with inital conditions $y(0)=y_0$ satisfying
$y(t) \in H$ for $t \in (0,c), y(t) \in \mathit{Inter}(H)$
for $t \in (0,\epsilon)$,
$y(c) \in F'$, 
and $y(t) \in \mathit{Inter}(H')$ for $t \in (c, c+\epsilon')$.
Let $\mathit{ExitSet}(H,E,F')$ denote the set of all such (targeted) points $y(c) \in F'$.

Define \emph{focal subset of $E$ on $H$ not leaving $H$}, $\mathit{Focal}(H,E,\emptyset)$, as the set of all points $y_0 \in E$ such that there exists a trajectory of a solution $y(t)$ of system~(\ref{eq:autsystem}) with initial conditions $y(0)=y_0$ satisfying $y(t) \in H$ for all $t>0$.
\end{definition}

Next we define the successor function for any pair $\langle H,E\rangle$ and subsequently the quantitative discrete approximation automaton.

\begin{definition}
\label{def:successors}
Let $E \in \mathit{EntrySets}_{\kappa}(H)$.  
Define the \emph{successors of} $\langle H,E \rangle$ as the set of pairs $\langle H',E'\rangle$ with $H'\in \mathit{Rect}(\mathcal{T}), E' \in \mathit{EntrySets}_{\kappa}(H')$ such that 
$$
\mathit{Succs}(\langle H, E\rangle) = \bigl\{ \langle H',E'\rangle \mid  H',E'  \mathrm{\ satisfy\ one\ of\ conditions\ 1.-3.\ below }  \bigr\}
$$
\begin{enumerate}
\item $H'\bowtie H$, $E \neq \emptyset$.
Denote $F'$ the facet of $H$ satisfying $F'=H\cap H'$. 
Let $n'=n$, if $E=H$, and $n'=(n-1)$, otherwise.
Moreover, $E'=\bigcup \mathit{Tiles}_{\kappa}\bigl(\mathit{ExitSet}(H,E,F')\bigr)$ and 
 $\lambda^{\kappa}_{n'}\bigl(\mathit{Focal}(H,E,\emptyset)\bigr) > 0$.

\item $H'=H$, $E\neq \emptyset$, and $E'=\emptyset$. Further, it holds that either $E \subseteq F$ and $\lambda^{\kappa}_{n-1}\bigl(\mathit{Focal}(H,E,\emptyset)\bigr) > 0$, or
$E = H$ and $\lambda^{\kappa}_{n}\bigl(\mathit{Focal}(H,E,\emptyset)\bigr) > 0$.
\item $H'=H$ and $E'=E=\emptyset$.
\end{enumerate}
\end{definition}

\begin{definition}[The Quantitative Discrete Approximation Automaton]
\label{def:qdaa}
Let $\kappa, \mathcal{B}$ be as above. 
The \emph{quantitative abstraction automaton} $QDAA_{\kappa}(\mathcal{B})$  of a biochemical system $\mathcal{B}$ with parameter $\kappa$ is a tuple $QDAA_{\kappa}(\mathcal{B})=\langle S, \mathcal{I}_C, \delta, p \rangle$, 
where
\itemize
\item the \emph{set of states} $S = \{\langle H, E\rangle \mid H \in \mathit{Rect}(\mathcal{T}), E \in \mathit{EntrySets}_{\kappa}(H)\},$
\item the \emph{set of initial conditions} $I_C = \left\{\langle H,H\rangle\mid H \in \mathcal{I}_C \right\}$,
\item the \emph{transition function} $\delta: S \rightarrow 2^S$ 
is defined as $\delta(\langle H, E \rangle) = \mathit{Succs}(\langle H, E \rangle)$,
\item the \emph{weight function} $p: S\times S \rightarrow  [0,1]$ is defined by the following expression, where $S=\langle H,E\rangle, S'=\langle H',E'\rangle$. Suppose $n'= n$, in case $E=H,$ and $n'=n-1$, otherwise.
{\small
$$p(S, S') =
\begin{cases}
1, & \mathrm{if\ }H=H',\,E=E'=\emptyset, \\
\vspace{3mm}
\dfrac{\lambda^*_{n'}\bigl(\mathit{Focal}(H,E,\emptyset)\bigr)}
{ \sum_{A \in \mathit{Facets}(H) \cup \{\emptyset\} }
\lambda^*_{n'}\bigl(\mathit{Focal}(H,E,A)\bigr)}, &
\mathrm{if\ }H=H',\,E\neq \emptyset,E'=\emptyset,\\
\dfrac{\lambda^*_{n'}(\mathit{Focal}(H,E,F'))}
{ \sum_{A \in \mathit{Facets}(H) \cup \{\emptyset\} }
\lambda^*_{n'}\bigl(\mathit{Focal}(H,E,A)\bigr)}, &
\mathrm{if\ }H \bowtie H', E' \subseteq F'=H\cap H', \\
0, & \mathrm{otherwise.}
\end{cases}$$}
\end{definition}

\begin{example} 
Assume the biochemical system from Figure~\ref{fig:example}.
See Figure~\ref{fig:qdaaint}~a) for an example of focal subsets described below.
Let $R=[0,2.5]\times[2.5,5]$ be a rectangle and let
$F_0 = \mathit{Facet}^{\top}_2(R), 
F_1 = \mathit{Facet}^{\top}_1(R),
F_2 = \mathit{Facet}^{\bot}_2(R),
F_3 = \mathit{Facet}^{\bot}_1(R).$
For the state $\langle R, F_0 \rangle$ the focal set of $F_1$ equals $F_0$, whereas $Focal(F_0)=Focal(F_2)=Focal(F_3)=\emptyset$.

Let $H=[0,2.5]\times[0,2.5]$ and 
$F = \mathit{Facet}^{\top}_1(R).$ 
For the state $\langle H, H \rangle$ 
the set $\mathit{Focal}(F)$ is the blue area inside $H$ and 
$\mathit{Focal}(\emptyset)$ is the yellow area.
All the solutions of the biochemical systems dynamics 
with initial conditions in $\mathit{Focal}(\emptyset)$ approach the yellow line of fixed points and stay in $H$ forever. All the solutions starting in the blue area leave $H$
in finite time through $F$. 

In the right part of Figure~\ref{fig:qdaaint} is the set of reachable states of the quantitative discrete approximation automaton (QDAA) obtained from the biochemical system described in Figure~\ref{fig:example} with initial conditions $\mathcal{I}_C=\{[0,2.5]\times[0,2.5]\}$. 

Let $H, R$ be the same as above. Let $S=[2.5,5]\times[0,2.5]$ and let $\mathcal{I}_C=\{H\}$.
The QDAA successor states of $\langle H,H \rangle$ are $\langle H, \emptyset \rangle$ 
(a selfloop state) and $\langle S,E \rangle$ (where $E$ denotes the $\kappa$-tiles approximation of the red segment in $\mathit{Facet}^{\bot}_1(S)$).
For $\kappa \rightarrow \infty$ the weights of these two transitions approach the area ratios of yellow and blue regions of $H$ respectively.
The only successor of $\langle H, \emptyset \rangle$ is (by definition) itself.
The state $\langle S,E \rangle$ has one successor $\langle S, \emptyset \rangle$, since
all the trajectories beginning in $E$ approach the line of fixed points and stay inside $S$ forever.

Therefore the set of concentrations reachable from initial rectangle $H$ is 
$[0,5]\times[0,2.5]$. See the rectangular abstraction transition system from Figure~\ref{fig:example} where the set reachable from $H$ is 
$[0,5]\times[0,2.5] \cup [2.5,5]\times[2.5,5],$ although there exists no trajectory of a solution of the biochemical systems dynamics that starts in $H$ and reaches a point inside $[2.5,5]\times[2.5,5].$

On the other hand, if $\kappa$ is too small, some behaviours of the system are not reflected in QDAA, because the set of $\kappa$-tiles corresponding to the entry set may be empty. With finer partition into $\kappa$-tiles smaller entry sets can be captured and approximation of the biochemical system by a QDAA is more realistic.
\end{example}

\begin{figure}[!h]
\vspace{-3mm}
\includegraphics[scale=.3]{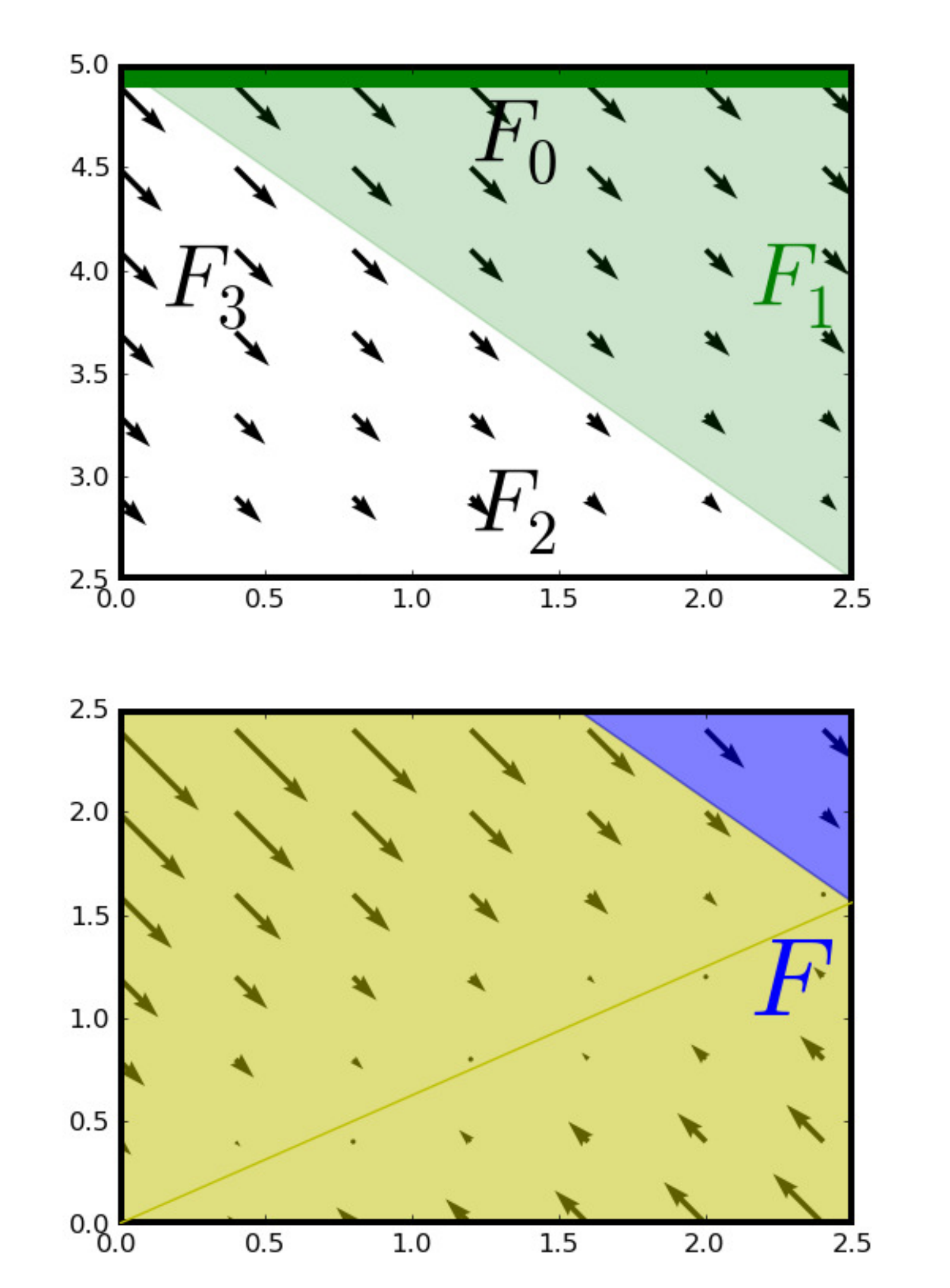}
\hspace{1cm}
\includegraphics[scale=.28]{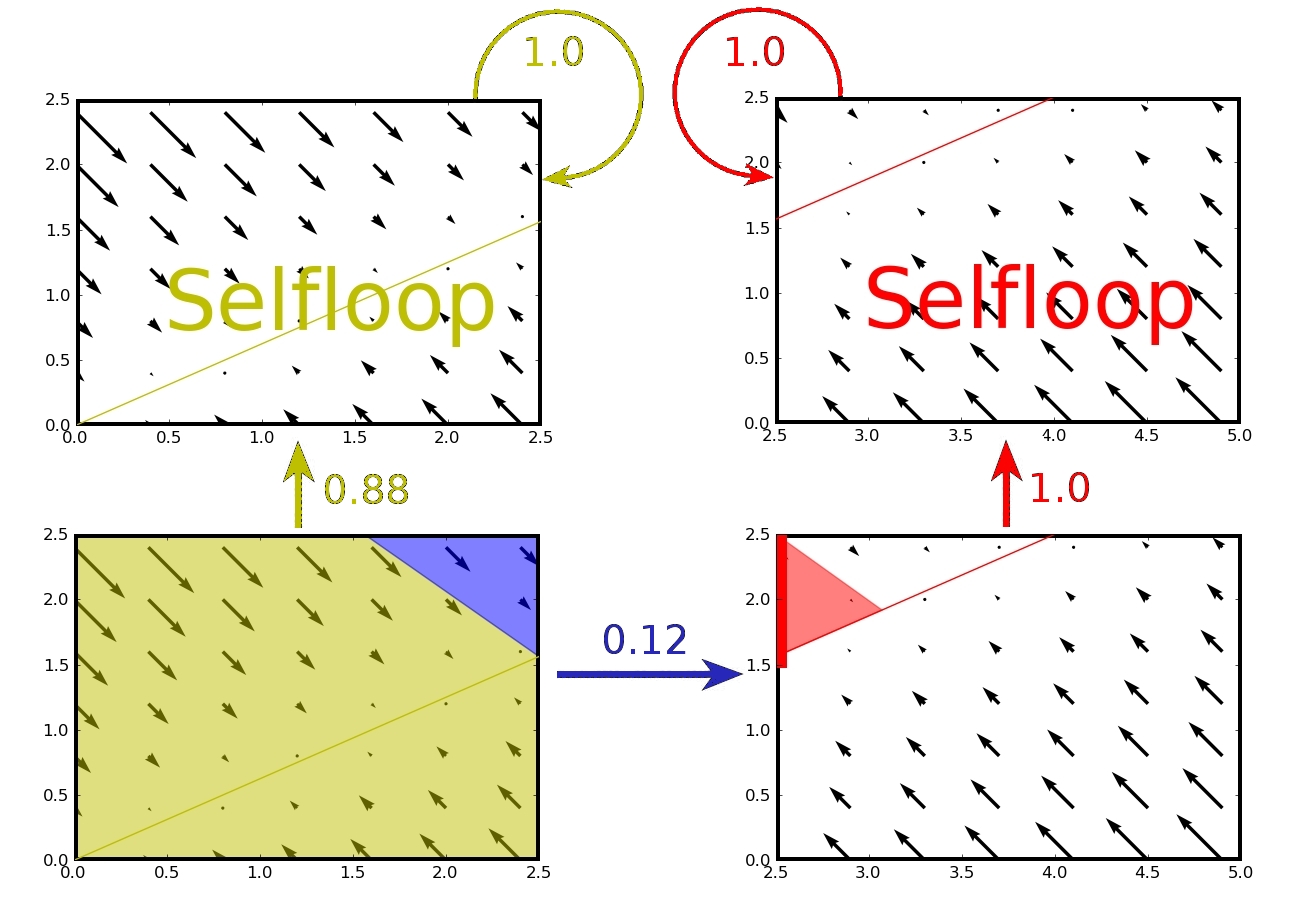}
\vspace{-4mm}
\newline
a)\hspace{5.5cm}b)
\caption{
a) Focal sets examples, 
b) QDAA example.} 
\label{fig:qdaaint}
\end{figure}

In the Theorem~\ref{thm:zero_measure} we ensure correctness of using the Lebesque measure  in Definition~\ref{def:qdaa}. We ensure that there is no non-zero volume entry set such that all trajectories from this set lead to a facet without entering the interior of a neighbouring rectangle.

Following notations and lemmas provide the background that is used in the proof of Theorem~\ref{thm:zero_measure}.

\begin{notation}Let us denote $B_{\epsilon}(x)$ the \emph{open sphere} with centre $x\in\mathbb{R}^n$ and radius $\epsilon>0$  ($B_{\epsilon}(x) = \{y \in \mathbb{R}^n | \left|x-y\right| < \epsilon\}$).
\end{notation}

The following theorems from mathematical analysis and differential topology (see for example~\cite{HirschDT}, \cite{HartmanODE}) will be used in the proofs.

\begin{lemma}[Peano]
\label{thm:c1_solutions} 
Let $f:\mathbb{R}^n \rightarrow \mathbb{R}^n$ be continuous on open set $E \subseteq \mathbb{R}^n$ such that $f$ possesses continuous first order partial derivatives.
Then for every $y_0\in E$, the initial value problem
$\dot{x} = f(x), x(0)=y_0$
has a unique solution $y(t)= \eta(t,y_0)$
and the unique solution has continuous first order derivative wrt the variable $t$ (is of class $C^1$) on its open domain of definition.
\end{lemma}

\begin{lemma}
\label{thm:smooth_diffeo_zero} The following statements about zero measure sets hold. (A set has Lebesgue outer measure zero iff it has Lebesgue measure zero.)
\enumerate
\item If $S\subset \mathbb{R}^n, \lambda_n^*(S)=0$ and $f:\mathbb{R}^n \rightarrow \mathbb{R}^n$ is a smooth map (with continuous first order derivatives), then $\lambda_n^*(f(S))=0$.
\item Being of zero measure is a diffeomorphism-invariant property of subsets of $\mathbb{R}^n$.
\end{lemma}

\begin{lemma}[A version of Fubini theorem]
\label{thm:compact_zero}
Let $U$ be a compact subset of $\mathbb{R}^n$. Denote by $\mathbb{R}_t^n$ the subset $\{t\}\times \mathbb{R}^{n-1} \subset \mathbb{R}^n$, and denote $U_t = U \cap \mathbb{R}_t^n$. If every set $U_t$ satisfies $\lambda_{n-1}^*(U_t) = 0$, then $\lambda_n^*(U) = 0$.
\end{lemma}

\begin{lemma}
\label{lem:zer:6}
Let $g:\mathbb{R}^n\rightarrow \mathbb{R}^n$ be a multiaffine function. Let $F \in \mathit{Facets}_i(H)$. The following statements hold.
\enumerate
\item If there exists a set $U \subseteq H$ such that $\lambda_n^*(U) > 0$ and $\forall x \in U : g(x)=0$, then $g \equiv 0$ on $H$ (and on $\mathbb{R}^n$).
\item If there exists a set $V \subseteq F$ such that $\lambda_{n-1}^*(V) > 0$ and $\forall x \in V : g(x)=0$, then $g \equiv 0$ on $F$ (and on the hypherplane $\mathbb{R}_i^n(c)$ containing $F$).
\end{lemma}
\begin{proof}
The proof of statement \emph{1.}  will be done 
 by mathematical induction wrt $n$.

For $n=1$ is $H$ a line segment, and the function $g$ is a linear function of one variable. Either there is one point (a set of measure zero in $\mathbb{R}^1$) where the function attains zero value, or the function is zero on the whole line.

For $n=2$, $H$ is a rectangle in plane. The multi-affine function of two variables is either zero on whole plane, or on two intersecting lines, or a hyperbolic curve or empty set. Therefore either the function is zero on the whole plane, or on a subset of plane of measure zero. Only the second case is compatible with $g$ being nonzero on a set of non-zero Lebesgue outer measure.
Assume that the statement holds for $1,2,\dotsc, n-1$ and let us prove it for $n$.
By Theorem~\ref{thm:compact_zero}.

\emph{Proof of 2.} is analogous, we can identify $\mathbb{R}_i^n(c)$ with $\mathbb{R}^{n-1}$ and assume that $g$ is a multiaffine function of $n-1$ variables (with $x_i=c$ constant).
\qed
\end{proof}

\begin{definition}
\label{def:hrouda}
Let $U \subset H$, let $I$ be an (bounded or unbounded) interval in $\mathbb{R}$. Define the \emph{$I$-trajectories unoin} of $U$ as the set
$$\Xi_I(U)\bigl\{ x(s) | s \in I, x(t)\mathrm{\ is\ a\ solution\ of\ system~(\ref{eq:autsystem})\ defined\ on\ }I, x(0) \in U\bigr\}.$$
\end{definition}

\begin{lemma}
\label{lem:zer:7}
Let $U\subset H$ be a $(n-1)$-dimensional closed disc in $H$ (i.e.
$$\exists i \in \{1,\dotsc,n\} \exists c \in \mathbb{R} \exists \delta > 0 \exists x \in H: U=\{y \in \mathbb{R}_i^n(c) | \left| x - y \right| \leq \delta \},$$
 we do not assume there is a facet $F$ of $H$ with $U \subseteq F$). 

Let  $\lambda_{n-1}^*(U)>0$ and let either $\forall y \in U: f(y)\cdot e_i > 0$ or $\forall y \in U: f(y)\cdot e_i < 0$. 

Then either $\exists U' \subseteq U$ with $\lambda_{n-1}^*(U')>0$ such that all the solutions of system~(\ref{eq:autsystem}) stay forever in $H$, or there exists $F' \in \mathit{Facets}(H)$ such that $\lambda_{n-1}^*(F' \cap \Xi_{[0,+\infty)}(U)) > 0$.

Moreover, in the second case,  $\lambda_{n-1}^*(F' \cap W) > 0,$ where $W$ is the set of points of the boundary of $H$ where the trajectories leave the rectangle $H$ for the first time (a subset of the connected component of $\Xi_{[0,+\infty)}(U) \cap H$ containing $U$).
\end{lemma}
\begin{proof}
Consider the case, when there is no $ U' \subseteq U$ with $\lambda_{n-1}^*(U')>0$ such that all the solutions of system~(\ref{eq:autsystem}) stay forever in $H$.

Since all compactly supported vector fields are complete, we can restrict ourselves on the multi-affine vector field $f$ defined on a compact neighbourhood of $H$ and assume that
flow of this (complete) vector field defines a one-parametric group of diffeomorphisms $Fl_t^f : (t,y_0) \rightarrow y(t),$ where $y(t)$ is the respective solution of system~\ref{eq:autsystem}. 

We consider a diffeomorphism $\varphi_U$ between $\Xi_{[0,+\infty)}(U) \cap H$ and a special subset $V$ of $U \times \mathbb{R}$.
We prove that for $U$ satisfying the assumptions of Lemma~\ref{lem:zer:7} the set $V$ is of nonzero measure. Therefore its image by diffeomorphism $\varphi$ must be of nonzero measure (due to Theorem~\ref{thm:smooth_diffeo_zero}). 
If on the other hand the set $\Xi_{(-\infty,0]}(F' \cap W) \cap H$ was of measure zero, its image in $\varphi_{F' \cap W}$ must be of zero measure also, but since
$\Xi_{[0,+\infty)}(U) \cap H$ is the union of finitely many such sets $\Xi_{(-\infty,0]}(F' \cap W) \cap H$, there has to exist $F'$ such that the set $\Xi_{(-\infty,0]}(F' \cap W) \cap H$ has nonzero measure.
\qed
\end{proof}

\begin{lemma}
\label{lem:zer:8}
Let $E \in \mathit{EntrySets}_{\kappa}(F,H), E \neq \emptyset$. 

Then either $\exists U' \subseteq E$ with $\lambda_{n-1}^*(U')>0$ such that all the solutions of system~(\ref{eq:autsystem}) stay forever in $H$, or there exists $F' \in \mathit{Facets}(H)$ such that
$$\lambda_{n-1}^*(F' \cap \Xi_{[0,+\infty)}(E \cap \mathit{Entry}(F,H))) > 0.$$ 

Moreover, in the second case, $\lambda_{n-1}^*(F' \cap W) > 0,$ where $W$ is the set of points of the boundary of $H$ where the trajectories leave the rectangle $H$ for the first time (a subset of the connected component of $\Xi_{[0,+\infty)}(E \cap \mathit{Entry}(F,H)) \cap H$ containing the set $E \cap \mathit{Entry}(F,H)$).
\end{lemma}
\begin{proof}
Consider again the case, when there is no $ U' \subseteq U$ with $\lambda_{n-1}^*(U')>0$ such that all the solutions of system~(\ref{eq:autsystem}) stay forever in $H$.

The situation is easier if there exists a point $x \in E$ with $f(x) \cdot \nu_H(F) < 0$. The multi-affine function $f$ is continuous, therefore there exists am open sphere containing $x$ such that $f(y) \cdot \nu_H(F) < 0$ for all points $y$ from this sphere. As a subset of the intersection of this sphere and $F$ there exists a closed disc like $U$ from the assumption of Lemma~\ref{lem:zer:7} and the statement of this lemma holds.

Otherwise the whole set $E \cap \mathit{Entry}(F,H)$ is a subset of $F$ with nonzero $(n-1)$-dimensional outer Lebesgue measure on which $f\equiv 0$.
By Lemma~\ref{lem:zer:6} $f\equiv 0$ on the whole $(n-1)$-dimensional hypher-plane containing $F$.

The idea of proof for this configuration is to use the definitorial property of $\mathit{Entry}(F,H)$, that insures existence of $\epsilon_z$ and a point $y \in \mathit{Inter}(H)$ reachable by a trajectory from a point $z$ in $E \cap \mathit{Entry}(F,H)$ in time $\epsilon_z$ with $f(z) \cdot \nu_H(F) < 0$ (that is implied by the properties of multi-affine function $f$). By continuous dependency on the initial conditions we find a disc around $z$ that cosists of points reachable from $E \cap \mathit{Entry}(F,H)$ (continuous dependence) and satisfies assumptions of Lemma~\ref{lem:zer:7}. Then we use Lemma~\ref{lem:zer:7} and this lemma is proved.
\qed
\end{proof}

\begin{theorem}
\label{thm:zero_measure}
Let $E \in \mathit{EntrySets}_{\kappa}(H), E \neq \emptyset$. Further, let $n'=n$, if $E=H$, and $n'=(n-1)$, otherwise.
Then
\begin{equation} \sum_{A \in \mathit{Facets}(H) \cup \{\emptyset\} }
\lambda^*_{n'}\bigl(\mathit{Focal}(H,E,A)\bigr) > 0,
\end{equation}
\end{theorem}

\begin{proof}
The proof will be divided into two parts.

First, for $E \subseteq F, E\neq \emptyset.$ Second, for $E=H$.

\emph{Proof~of~1.} The statement of theorem is obtained by using Lemma~\ref{lem:zer:8}.

\emph{Proof~of~2.}
Let $E=H$. In the case when all the trajectories with initial points in $H$ stay forever in $H$ 
the inequality $\lambda^*_{n}\bigl(\mathit{Focal}(H,E,\emptyset)\bigr) > 0$ holds.

In the case when there exists a point $x_0 \in H$ and a trajectory of a solution $x(t)$ of system~(\ref{eq:autsystem}) with $x(0)=x_0, \exists c, \epsilon>0: \forall t\in [0,c], x(t)\in H, \forall t \in (c, c+\epsilon) x(t) \notin H.$
Let us denote $y$ the point $x(c+\frac{\epsilon}{2})\notin H$. From Theorem~\ref{thm:continuous_dependence} (continuous dependency on initial conditions) there exists an $n$-dimensional sphere $B$ with centre $x_0$ such that trajectories of all solutions of~(\ref{eq:autsystem}) with initial point in the sphere $B$ leave the rectangle $H$ and continue into a neighbourhood of point $y$. The intersection $H \cap B$ surely contains a disc satisfying assumptions of Lemma~\ref{lem:zer:7} and therefore there exists $F' \in \mathit{Facets}(H)$ such that the intersection of $\Xi_{[0,+\infty)(B)}$ and $F'$ is of nonzero measure, i.e. it has a subset of nonzero $(n-1)$-dimensional measure in the interior of the facet $F'$. Then there exists an open sphere $B'$ contained in $B$ such that all trajectories of solutions with initial points in $B'$ leave the rectangle $H$ through the interior of facet $F'$, and $\lambda^*_{n}\bigl(\mathit{Focal}(H,E,F')\bigr) \geq \lambda_{n}^*(B') > 0$. 
\qed
\end{proof}



\begin{theorem}
\label{thm:mc}
The quantitative abstraction automaton $\mathrm{QDAA}_{\kappa}(\mathcal{B})$ of a biochemical system $\mathcal{B}$ is a discrete time Markov chain.
\end{theorem}

\begin{proof}
The number of states is finite, bounded by $\left|\mathit{Rect}(\mathcal{T})\right|
\Big(2 + 2n \cdot \big( 2^{\kappa^{n-1}}-1\big)\Big)$.

Sum of probabilities of transitions from one state $\langle H,E \rangle$ equals
$1$ for $E=\emptyset$ and 
$\frac{\sum_{A \in \mathit{Facets}(H) \cup \{\emptyset\} }
\lambda^*_{n'}\bigl(\mathit{Focal}(H,E,A)\bigr)}{
\sum_{A \in \mathit{Facets}(H) \cup \{\emptyset\} }
\lambda^*_{n'}\bigl(\mathit{Focal}(H,E,A)\bigr)}$. The later sum equals $1$, whenever its denominator is nonzero (it is the case because of Theorem~\ref{thm:zero_measure}).
The probabilities of transitions from a given state are independent of previous states of the automaton.
\qed
\end{proof}

Finally, we provide a theorem suggesting that for sufficiently large values of parameter $\kappa$, the rectangular $\kappa$-grid measure of a bounded set $X$ contained in the phase space of biochemical system approaches its Lebesque outer measure (Theorem~\ref{thm:kappa_lebesgue}). Following lemmas and definition will be used in the proof of this theorem.

Let $\mathcal{B}=\langle n,f,\mathit{T}, \mathit{I}_C\rangle$ be a biochemical system, $H \in \mathit{Rect}(\mathcal{T}), F \in \mathit{Facets}(H),$ and $\kappa \in \mathbb{N}$ throughout this section.

\begin{definition}
Let $X\subseteq \prod_{i=1}^n [\min (T_i), \max (T_i)].$ Define $\lambda_n^{\kappa}(X)$ as the sum 
$$\sum_{H \in \mathit{Rect}(\mathcal{T})} \lambda_n^{\kappa}(X \cap H).$$
Analogously define $\lambda_{n-1}^{\kappa}(X)$ as the sum 
$$\sum_{F \in \bigcup_{H \in \mathit{Rect}(\mathcal{T})} \mathit{Facets}(H)}  \lambda_{n-1}^{\kappa}(X \cap F).$$
\end{definition}

\begin{lemma}
\label{lem:kap:1}
Let $m \in \{n-1,n\}$, let $J$ be an $m$-dimensional interval in $n$-dimensional space, $J \subseteq \prod_{i=1}^n [\min (T_i), \max (T_i)].$ Then 
$$
\lim_{\kappa\rightarrow \infty} \lambda_m^{\kappa}(J) = \lambda_m^*(J)
$$
\end{lemma}

\begin{proof}
For given $\kappa$ there are at most $2n\cdot \kappa^{n-1}$ distinct tiles $R$ satisfying $R \nsubseteq J$ and $R \cap J \neq \emptyset$.

The difference between $\lambda_m^{\kappa}(J)$ and $\lambda_m^*(J)=\mathit{vol}(J)$ is bounded by 
$$\dfrac{1}{2} \sum vol(R_i) = \dfrac{1}{2} 2n\cdot \kappa^{n-1} \mathit{vol}(R_1) = $$
$$= n \kappa^{n-1} \dfrac{V}{\kappa^n} = \dfrac{nV}{\kappa}$$
The expression approaches zero as $\kappa \rightarrow \infty$.
\qed
\end{proof}

\begin{lemma}
\label{lem:kap:2}
Let $M$ be a positive integer number ($M < \infty$).
Let $U = \bigcup_{i=1}^M I_i$ be a union of $M$ $n$-dimensional bounded rectangles, subsets of $\prod_{i=1}^n [\min (T_i), \max (T_i)]$. Let this the intersection of every two intervals be of zero Lebesgue outer measure. Then
$$
\lim_{\kappa\rightarrow \infty} \lambda_n^{\kappa}(U) = \lambda_n^*(U)
$$
\end{lemma}

\begin{proof}
There are maximally $M \cdot 2n \kappa^{n-1}$ distinct tiles $R$ satisfying $R \nsubseteq U$ and $R \cap U \neq \emptyset$.

The difference between $\lambda_m^{\kappa}(U)$ and $\lambda_m^*(U)=\mathit{vol}(U)$ is bounded by 
$$M \dfrac{1}{2} \sum \mathit{vol}(R_i) = M \dfrac{1}{2} 2n\cdot \kappa^{n-1} \mathit{vol}(R_1) = $$
$$= M n \kappa^{n-1} \dfrac{V}{\kappa^n} = M \dfrac{nV}{\kappa}$$
The expression approaches zero as $\kappa \rightarrow \infty$.
\qed
\end{proof}

\begin{lemma}
\label{lem:kap:3}
Let $X$ be a (bounded) subset of $\mathbb{R}^n$ and let $J_1, J_2, \ldots$ be a (countable) sequence of intervals such that $X \subseteq \bigcup_{j=1}^{\infty} J_j$. Then there exists a sequence $I_1, I_2, \ldots$ of intervals with $\lambda_n^*(I_i \cap I_j) = 0$ for each pair $i\neq j$ such that $\bigcup_{i=1}^{\infty} I_i = \bigcup_{j=1}^{\infty} J_j$.
\end{lemma}

\begin{proof}
The sequence of pairs of intervals from the first sequence is also countable. 

Replace every two rectangles that overlap by finitely many new non-overlapping rectangles. 
\qed
\end{proof}

\begin{lemma}
\label{lem:kap:4}
Let $X\subseteq \prod_{i=1}^n [\min (T_i), \max (T_i)]$ be a bounded subset of $\mathbb{R}^n$ with lebesgue outer measure $\lambda^*_n(X)=r \in \mathbb{R}^+_0$. 

Let $\epsilon > 0$.
Let $J_1, J_2, \ldots$ be a cover of $X$ by intervals such that
$\sum_{i=1}^{\infty} \mathit{vol}(J_i) - r < \epsilon$.

Let $\delta > 0$ be a positive real number. 
Then there exists a number $k_{\delta}$ such that 
$\sum_{i=1}^{k_{\delta}} vol(J_i) \in (r-\delta, r+\delta)$.

(That means $\sum_{i=k_{\delta}}^{\infty} vol(J_i) < \delta$.)
\end{lemma}

\begin{proof}
The sequence of partiall sums of the sequence of volumes $(vol(J_i))_{i=1}^{\infty}$ converges to $r$. That implies the the existence of such $k_{\delta}$.
\qed
\end{proof}

\begin{corollary}
\label{cor:kap:1}
Let $X$ and $\delta$ be the same as in the previous lemma.
Let $I_1, I_2, \ldots$ be a cover of $X$ by countably many intervals with $\lambda_n^*(I_i \cap I_j) = 0$ for each pair $i\neq j$.
Let $k_{\delta}$ be as in the previous lemma.
Then there exist two subsets of $X$ denoted $X'$ and $X''$ such that $X = X' \cup X''$, $\lambda^*_n(X' \cap X'')=0$, $X'' \subseteq \bigcup_{i=1}^{k_{\delta}} I_i$ and $X' \subseteq \bigcup_{i=k_{\delta}}^{\infty} I_i$.
\end{corollary}

\begin{proof}
The existence of $X'$ and $X''$ is obvious (can be defined as $X'' = \bigcup_{i=1}^{k_{\delta}} I_i \cap X$ and  $X' = \bigcup_{i=k_{\delta}}^{\infty} I_i \cap X$ respectively), there remains the proof of the equality $\lambda^*_n(X' \cap X'')=0$.

The set $X' \cap X''$ contains only countably many intersections of intervals $I_1, I_2, \ldots$, and these intersection are of measure zero, therefore their union has also measure zero, $\lambda_n^*(X' \cap X'')=0$.
\qed
\end{proof}

\begin{lemma}
\label{lem:kap:5}
Let $X$ be a subset of $ \prod_{i=1}^n [\min (T_i), \max (T_i)]$ with Lebesgue outer measure $\lambda_n^*(X) = r$. Then for every $\kappa$ the following statement is true:
$$\lambda_n^{\kappa}(X) \leq 2 r$$
\end{lemma}
    
\begin{proof}
For any tile $J$ to be counted into the $\kappa$-grid measure of a set, the inequalities 
$ \lambda_n^*(X \cap J) \geq \frac{1}{2} \mathit{vol}(J) \Leftrightarrow  2 \lambda_n^*(X \cap J) \geq \mathit{vol}(J)$
must hold. Therefore the set $X$ with Lebesgue outer measure $r$ can saturate at most a set of tiles of the overall volume $2r$.
\qed
\end{proof}

\begin{theorem} 
\label{thm:kappa_lebesgue}
Let $X \subseteq H \in \mathit{Rect}(\mathcal{T})$
(or more generally $X\subseteq \prod_{i=1}^n [\min(T_i),\max(T_i)]$, where $\langle T_1,\dotsc, T_n\rangle = \mathcal{T}$ is the partition of $\mathcal{B}$).
Then
\begin{equation}
\lim_{K\rightarrow \infty} \lambda^{\kappa}_n(X) = \lambda^*_n(X).
\end{equation}
\end{theorem}

\begin{proof}
The proof uses the definition of outer Lebesgue measure and our aim is to prove that for every $\epsilon > 0$ there exists such $\kappa$ that the difference $|\lambda^{\kappa}_n(X) - \lambda^*_n(X)| < \epsilon$. 

Let 
$\epsilon > 0$ be a positive real number, denote $r=\lambda^*_n(X)$.

For $X$ in the type of interval the proof is easy (Lemma~\ref{lem:kap:1}).

For general $X$ consider sufficiently accurate (with less than $\frac{\epsilon}{2}$ difference of sum of volumes and $r$) cover of $X$ with countable collection of intervals whose interiors do not intersect (Lemma~\ref{lem:kap:3}).
 
There exist two subsets of $X$ denoted $X'$ and $X''$ such that $X=X' \cup X'', \lambda^*_n(X' \cap X'') = 0$, such that the collection of intervals can be divided into a finite part $I_1,\ldots, I_k$ and the remainder $I_{k+1},\ldots$ such that the sum of volumes of the remainder is sufficiently small (less than $\frac{\epsilon}{8}$) (Corollary~\ref{cor:kap:1}). 

For a bounded set $Y$ with $\lambda^*_{n}(Y)=s$ the inequality $\lambda^{\kappa}_n(Y) \leq 2s$ holds for every $\kappa$ (Lemma~\ref{lem:kap:5}).

For the finite set of $k$ intervals we can find such $\kappa$ that $|\sum_{i=1}^k \lambda^{\kappa}_n(I_i) - \sum_{i=1}^k\lambda^*_n(I_n)| \leq \frac{\epsilon}{4}$ (applying Lemma~\ref{lem:kap:1} finitely many times, taking the maximal of obtained values of $\kappa$)
and the overall difference of $\lambda^{\kappa}_n(X)$ and $\lambda^*_n(X)$ is less than $\epsilon$.
%
%
%
%
\qed
\end{proof}

Note that the result applies also to the case with $X\subseteq F \in \mathit{Facets}(H)$ and $\lambda^{\kappa}_{n-1}, \lambda^*_{n-1}$.

\section{Algorithm}
\label{sec:algorithm}
This section introduces procedures for obtaining the reachable state space of the quantitative discrete approximation automaton.
Algorithm~\ref{alg:bfs} is a procedure of computing the set of reachable states.
Algorithm~\ref{alg:getsuccs} describes the computation of transitions from one state (i.e. successors) together with their weights using numerical simulations.

The procedure of computing reachable state space (Algorithm~\ref{alg:bfs}) is based on breadth first search. 
States corresponding to initial conditions of the biological system are enqueued first and a list of states already visited is maintained.
The computation is always finite, because there are only finitely many possible states of the automaton and each of them can be at most once added and after the computation of its successors removed from the queue.

\begin{algorithm}
 \caption{Computing the set of reachable states}
 \label{alg:bfs}
\begin{algorithmic}[1]
{\scriptsize
    \REQUIRE $\mathcal{B}=(n,f,\mathcal{T},\mathcal{I}_C)$, $\kappa \in \mathbb{N}$
    \ENSURE $\mathrm{Reachable} = \mathrm{set\ of\ all\ reachable\ states\ of\ the\  automaton\ }QDAA_{\kappa}(\mathcal{B})$
    \STATE Reachable $\leftarrow \emptyset$
    \FORALL {$H \in \mathcal{I}_C$}
        \STATE s $\leftarrow \langle H, H \rangle$
        \STATE Reachable $\leftarrow$ Reachable $\cup \{ s\}$\\
        \STATE Queue.pushBack($s$)
    \ENDFOR
    \WHILE { Queue $\neq \emptyset$ }
        \STATE $s \leftarrow$ Queue.firstElement \\
        \STATE $A \leftarrow$ getSuccessors($s$)\\
        \FORALL {$a \in A$}
            \IF {$a \notin $ Reachable} 
                \STATE Reachable $\leftarrow$ Reachable $\cup \{ a \}$\\
                \STATE Queue.pushBack($a$)
            \ENDIF
        \ENDFOR
    \ENDWHILE
    \RETURN Reachable
}
\end{algorithmic}
\end{algorithm}

Computation of the successors (Algorithm~\ref{alg:getsuccs}) of one state requires determining the rectangles and the entry sets of the successors and weights of the transitions. 
This can be done approximately using numerical simulations. 
We sample the entry set of the state and perform numerical simulations with the sampled points as initial conditions and the dynamics of the given biological system as the vector field. 
For each simulated trajectory we watch whether it leaves the rectangle before given maximal time interval elapses. If this is the case then the location of the exit point through which the trajectory leaves the rectangle is of interest.

Entry sets of the successor states are also determined within Algorithm~\ref{alg:getsuccs}.
If the successor is a selfloop state the entry set is empty.
For a neighbouring rectangle successor with one common facet the entry set  is computed using the exit points locations and more numerical simulations. 
From the set of exit points in a facet we can estimate the set of $\kappa$-tiles of the facet that surely have nonempty intersection with the exit set.
It remains to decide in which of the $\kappa$-tiles the intersection of the tile with the exit set takes at least one half of the volume of the tile.

To this end we use numerical simulations and the fact that for an autonomous system of ODEs $\dot{x} = f(x)$ with a solution $x(t)$ the function $x(-t)$ is a solution of autonomous system $\dot{x} = -f(x)$.
For determining whether to include a $\kappa$-tile in the entry set of a successor state, we sample the tile and perform numerical simulations of the trajectories of system $\dot{x} = -f(x)$. If more than one half of the simulated trajectories go through the rectangle and the entry set of the original state, then the $\kappa$-tile is included in the entry set of successor state, otherwise the $\kappa$-tile is not included.

Weights of the transitions correspond to portions of the set of  performed simulations that leave the rectangle to the respective neighbouring rectangles.
Weight of the transition from the state to the so-called selfloop state with the same rectangle is determined as the portion of trajectories that do not leave the rectangle in given maximal time interval.

\begin{algorithm}[!h]
 \caption{Procedure getSuccessors}
\begin{algorithmic}[1]
\label{alg:getsuccs}
{\scriptsize
    \REQUIRE $\mathcal{B}=(n,f,\mathcal{T},\mathcal{I})$, $\kappa,M \in \mathbb{N}$, 
    $H \in \mathit{Rect}(\mathcal{T})$, $E \in \mathit{EntrySets}(H)$
    \ENSURE $\mathrm{Successors} = \mathit{Succs}_{\kappa}(\langle H,E \rangle)$
    \IF{$E=\emptyset$}
        \STATE Successors $\leftarrow \{\langle H, \emptyset \rangle \}$
        \RETURN Successors\\
    \ENDIF

    \STATE $A \leftarrow$ set of $M$ random points in $E$
    \STATE ExitPoints $\leftarrow \emptyset$ 
    \STATE StaysInside $\leftarrow 0$
    \FORALL {$x_0 \in A$}
        \STATE simulate trajectory from $x_0$ until it leaves $H$ through a point $x_1$ or given time elapses
        \IF {$x_1$ exists}
            \STATE ExitPoints $\leftarrow$ ExitPoints $\cup \{x_1\}$
        \ELSE
            \STATE StaysInside $\leftarrow$ StaysInside $+ 1$
        \ENDIF
    \ENDFOR
    \FORALL {F$ \in \mathit{Facets}(H)$, F$=H\cap H'$}
        \IF {ExitPoints $\cap$ F $\neq \emptyset$}
            \STATE EntryTiles $\leftarrow \{ Z \in \mathrm{Tiles}_{\kappa}(F) \mid Z \cap$ ExitPoints $\neq \emptyset\}$ 
            \FORALL {Z $\in$ EntryTiles}\label{line:back_begin}
                \STATE $B \leftarrow$ set of $M$ random points in Z
                \STATE RealPointsCount $\leftarrow 0$
                \FORALL {$y_0 \in B$}
                    \STATE simulate trajectory from $y_0$ until it leaves $H$ through a point $y_1$ or given time elapses
                    \IF {$y_1 \in E$}
                        \STATE RealPointsCount $\leftarrow$ RealPointsCount $+1$
                    \ENDIF
                \ENDFOR
                \IF {RealPointsCount $< \frac{M}{2}$}
                    \STATE EntryTiles $\leftarrow$ EntryTiles $\setminus \{$Z$\}$
                \ENDIF
            \ENDFOR \label{line:back_end}
        \ENDIF
        \IF {EntryTiles $\neq \emptyset$}
            \STATE Successors $\leftarrow$ Successors $\cup \langle H', EntryTiles \rangle$
            \STATE Weight[$\langle H, E \rangle$][$\langle H, EntryTiles \rangle$] $\leftarrow \dfrac{|\mathrm{ExitPoints} \cap \mathrm{F}|}{|A|}$ 
        \ENDIF
    \ENDFOR

    \RETURN Successors
}
\end{algorithmic}
\end{algorithm}
%

Performing the backward simulations (lines~\ref{line:back_begin}--\ref{line:back_end} of Algorithm~\ref{alg:getsuccs}) can be switched off. The resulting transition system differs from the QDAA in the entry sets, that can be larger. 
Difference of the outputs can be seen on Figure~\ref{fig:bayramov}. The algorithm with backward simulations computes the QDAA and for $(\kappa \rightarrow \infty)$ approaches the real behaviour of the solutions of dynamics ODE system. On the other hand the algorithm without backward simulations overapproximates the entry sets, therefore the transitions are included even if the entry set of a state is smaller than half of one $\kappa$-tile. Both options still lead to automatons with reachable states whose rectangles are included in the set of reachable rectangles of the rectangular abstraction with the same initial rectangles.

The worst case complexity of the algorithms follows.
There are at most $k^{n}$ rectangles in the phase space of the biochemical system, where $k$ is the maximal number of thresholds on one variable. 
The maximal number of states of QDAA of the form $\langle H, E \rangle$ for a fixed rectangle $H$ is 
$2n \cdot \big( 2^{\kappa^{n-1}}-1\big)$, where $n$ is the dimension of the biochemical system.
For the average numbers of visited different states of QDAA with the same rectangle encountered while analysing our evaluation models see the line labeled $\varrho$ in Table~\ref{tab:results}.
Complexity of the computation of successors of a given state depends on the dimension of the system, the $\kappa$ parameter and on the number of simulations $M$ used per one tile. In the worst case when all the tiles are examined (either as a part of entry set or potential exit set) there are
$2n \cdot \kappa^{n-1} \cdot M$ simulations.

Visualization of the state space of QDAA involves highlighting the borders of the rectangles $H$ such that there is at least one state $\langle H,E\rangle$ visited during the computation. The intensity of the fill colour of a rectangle $H$ is calculated proportional to the sum of weights of all possible paths from initial set $\mathcal{I}_C$ to the first appearance of states with $H$ as the rectangle. The weight of a finite path is obtained as the product of weights of the subsequent transitions in the path. The sum is always between zero and one.


\subsection{Entering and leaving conditions}
The computation of successor states proceeds in two steps. 

First, the probabilities of potential successors are computed for all $\kappa$-tiles of $E$ and summed up to get the probabilities of successors for the whole set $E$. 
For each $\kappa$-tile of $E$ several numerical simulations of a solution of system~(\ref{eq:autsystem}), with initial point $x(0)$ placed randomly in the tile, are performed. If the computed trajectory satisfies entering and leaving conditions (Definition~\ref{def:etacondit}) and leaves the box $H$ or the maximal time interval elapses, the number of trajectories leaving $H$ through the particular facet (resp. the number of trajectories assumed to stay forever in $H$ and leading to the transition $\langle H, E \rangle \rightarrow \langle H, \emptyset \rangle$) is increased. 

Second step of the algorithm takes into account the rectangles $H' \bowtie H$ with nonzero probability of transition $\langle H, E\rangle\rightarrow\langle H', \mathrm{yet\ unknown\ } E'\rangle$ computed in the first step and determines the entry sets $E'$ of this successors. Denote $F=H'\cap H$. For every $\kappa$-tile of $F$ the algorithm decides if the tile is a subset of $E'$. 

Replacing $\mathit{Entry}(F,H)$ with more easily computed set $\mathit{Entry'}(F,H)$ (defined in Definition~\ref{def:etacondit}) of points satisfying the entering condition in the definition of QDAA 
does not lead to a nonzero difference in the values of transitions probabilities due to Theorem~\ref{thm:epsilon_eta}.

\begin{definition}
Let $F \in \mathit{Facets}_i(H)$ be a facet of the form 
$F= [a_1,b_1]\times \ldots \{c\} \ldots \times[a_n,b_n].$ 
Define the \emph{normal vector} $\nu_H(F)$ to the facet $F$ with respect to the rectangle $H$ as the vector
$-e_i$ for $c=a_i$ and $e_i$ for $c=b_i$, where $e_i$ denotes the $i$th vector of the standard basis of $\mathbb{R}^n$ (i.e. the vector orthogonal to $F$ and pointing outside from $H$).
\end{definition}



\begin{definition}
\label{def:etacondit}
Consider rectangle $H, F,F'\in \mathit{Facets}(H)$, a mutiaffine vector field $f$, a solution $x(t)$ of the system~(\ref{eq:autsystem}) and $r>0$ satisfying
$x(0) \in F, x(r) \in F', \forall t \in (0,r): x(t) \in H.$ We say that 
$x(t)$ satisfies the \emph{entering (resp. leaving) condition} with respect to $H$ and $f$, if 
$f(x(0)) \cdot \nu_H(F) < 0)$ (resp. $f(x(r)) \cdot \nu_H(F') > 0$). 

Define $\mathit{Entry}'(F,H)= \{x \in F | \nu_{H}(F) \cdot f(x) < 0\}$.
\end{definition}
Our next step is to show (in Theorem~\ref{thm:epsilon_eta}) that the focal set of $\mathit{Entry'}(F,H) \setminus \mathit{Entry}(F,H)$ is of measure zero, thus the set difference is insignificant for our volume based notions.

Now we will extend our definition of $\mathit{Focal}(H,E,F')$ set of points from which the trajectories leave the box through a facet $F'$  (Definition~\ref{def:successors})  to $\mathit{Focal}(H,E,X)$ for an arbitrary subset $X$ of facet $F'$. 
\begin{definition}
Let $E \in \mathit{EntrySets}_{\kappa}(H), E \neq \emptyset, F'\in Facets(H)$. Let $X\subseteq F'$. 
Define $\mathit{Focal}(H,E,X)$ the set of all points $y_0 \in E$ such that there exist $\epsilon, c, \epsilon' > 0$ and a solution 
$y(t)$ of the system~(\ref{eq:autsystem}) satisfying 
$y(t) \in H$ for $t \in (0,c), y(t) \in \mathit{Inter}(H)$
for $t \in (0,\epsilon)$,
$y(c) \in X$, 
and $y(t) \in \mathit{Inter}(H')$ for $t \in (c, c+\epsilon')$.
\end{definition}

\begin{theorem}
\label{thm:epsilon_eta}
Let $E\in \mathit{EntrySets}_{\kappa}(H), E\neq \emptyset, F' \in \mathit{Facets}(H)$.
Let $V_+= \{x \in \mathit{ExitSet}(H,E,F') | \eta_{F'} \cdot f(x) > 0\}$ and let $V_0 = \mathit{ExitSet}(H,E,F') \setminus V_+$.
Then
$\lambda^*_{n'}\bigl( \mathit{Focal}(H,E,V_0) \bigr) = 0,$ where $n'$ denotes $n$ for $E=H$, and $(n-1)$ otherwise.
\end{theorem}

\begin{proof}
First, we have to observe that trajectories from such initial points have to leave the rectangle through a set of measure zero (in fact it is a zero set of a multi-affine polynomial).

Then the proof is simillar to proof of Lemma~\ref{lem:zer:7}.
\qed
\end{proof}

\begin{remark}
\label{rem:etaeps1}
Theorem~\ref{thm:epsilon_eta} implies that replacing $\mathit{Entry}(F,H)$ with $\mathit{Entry'}(F,H)$ in the definition of QDAA 
(recall that these can be only entry sets of the successor states, not of the initial states $\langle H, H \rangle$), does not lead to a nonzero difference in the values of transitions probabilites.
\end{remark}
\begin{remark}
\label{rem:etaeps2}
Deciding if a point $x$ is an element of $\mathit{Entry}'(F,H)$ is straightforward (compared to checking the $\epsilon$-condition of Definition~\ref{def:entry}), and $\mathit{Entry}'(F,H) \subseteq \mathit{Entry}(F,H)$. 
\end{remark}
\begin{remark}
\label{rem:etaeps3}
Moreover, there is the following symmetry property. 
Let $F,F' \in \mathit{Facets}(H), x_0 \in F, x_1 \in F'$. 
There is a solution $y(t)$ of system~(\ref{eq:autsystem})  satisfying 
$y(0)=x_0, \exists t_1: y(t_1)=x_1, f(y(0))\cdot \nu_H(F) < 0, f(y(t_1))\cdot \nu_H(F')> 0$ and $y(t) \in H$ for $t\in(0,t_1)$,  if and only if 
there is a solution $x(t)$ of system $\dot{x}(t)=-f(x)$  satisfying 
$x(0)=x_1, \exists t_1: x(t_1)=x_0, -f(x(0))\cdot \nu_H(F') < 0, -f(x(t_1))\cdot \nu_H(F) > 0$ and $x(t) \in H$ for $t\in(0,t_1)$.
\end{remark}

\section{Evaluation and Case Study}
\label{sec:case_studies}
In this section the state spaces of several biological models (of dimensions two, four and seven) are explored.
Using our prototype implementation of the algorithms from Section~\ref{sec:algorithm} implemented in C++, we evaluate our approach on two exemplary biochemical systems. 
Additionally, we provide a case study held on a biochemical pathway studied in \emph{E. coli} and compare the reachability results of the case study and one of the smaller models with results obtained using the rectangular abstraction approach.

Before we proceed with the models, let us introduce several terms useful for the evaluation. For a biochemical system $\mathcal{B}=\langle n,f,\mathcal{T},\mathcal{I}_C\rangle$ we denote $\mathcal{R}(\mathcal{I}_C)\subseteq \mathit{Rect}(\mathcal{T})$ the \emph{set of all rectangles reachable from initial set $\mathcal{I}_C$}. 
For each $H\in \mathit{Rect}(\mathcal{T})$ we denote $\mathit{mem}(H)$ the subset of $\mathcal{R}(\mathcal{I}_C)$ consisting of all states reachable from the initial set with $H$ as rectangle, the so-called \emph{memory of the rectangle}~$H$, 
$\mathit{mem}(H)= \{ \langle R,E\rangle \in \mathcal{R}(\mathcal{I}_C) \mid  R=H\}$. Further we denote $\varrho$ the average number of memory states (cardinality of $\mathit{mem}(H)$ averaged over all $H\in\mathcal{R}(\mathcal{I}_C)$).
The number of QDAA states representing the memory of a rectangle is in the worst case equal to the number of all its possible entry sets. 
However, the actual values of $\varrho$ in our examples are much smaller (see Table~\ref{tab:results}). 

Let us focus on the effect of parameter $\kappa$ on cardinality of $\mathcal{R}(\mathcal{I}_C)$ and on $\varrho$. Expected behaviour of the approximation is the following. Every facet is divided into $\kappa^{n-1}$ tiles. A tile is included in the entry set $E$ of some reachable state $\langle H, E \rangle$ if the focal subset $\mathit{Focal}(H,E,A)$ fills at least half of the volume of the tile. For higher values of $\kappa$, the set $\mathit{Tiles}_{n'}^{\kappa}(\mathit{Focal}(H,E,A))$ better approximates the set $\mathit{Focal}(H,E,A)$ because of the higher $\kappa$-grid resolution. Thus with increasing $\kappa$, the quantitative information denoting the probability of reaching states in $\mathcal{R}(\mathcal{I}_C)$ can be computed more precisely. We demonstrate that on models examined below.

\subsection{Oscillatory model}
First, we consider a $2$-dimensional model which is a variant of Lotka-Volterra model with oscillatory behaviour.
The oscillatory model has the form of the following multi-affine system:



$$\begin{array}{l}
 \frac{dX}{dt} =  5\cdot X - 1\cdot X\cdot Y\\[2mm]
 \frac{dY}{dt} =  0.4\cdot X\cdot Y - 5.4\cdot Y
\end{array}
$$


We consider the following partition and initial conditions for this model:

$$\begin{array}{l}
T_X=\{i|i\in\langle 0,30\rangle\subseteq \mathbb{N}_0\}\\[2mm]
T_Y=\{i|i\in\langle 0,12\rangle\subseteq \mathbb{N}_0\}
\end{array}$$

$$\mathcal{I}_C: X\in \langle 20,21\rangle,
Y\in \langle 5,6\rangle
$$

  Results achieved on our implementation are presented in Table~\ref{tab:results} and visualized in Figure~\ref{fig:bayramov}. Black rectangles denote the initial set. 

\subsection{Enzyme kinetics}
Similarly, we examined a $4$-dimensional model of basic enzyme kinetics based on the following set of reactions:

$$\begin{array}{rl}
S + E &\ra{k_1} ES\\[2mm]
ES &\ra{k_2} S + E\\[2mm]
ES &\ra{k_3} P + E
\end{array}$$

The corresponding multi-affine ODE model considered in the paper is the following:

$$\begin{array}{l}
\frac{dS}{dt} = -0.01\cdot S\cdot E + 1\cdot ES\\[2mm]
\frac{dE}{dt} =  1\cdot ES -0.01\cdot S\cdot E + 1\cdot ES\\[2mm]
\frac{dES}{dt} = -1\cdot ES + 0.01\cdot E\cdot S -1\cdot ES\\[2mm]
\frac{dP}{dt} = 1\cdot ES
\end{array}
$$

We consider the following partition and initial conditions for this model:

$$\begin{array}{l}
T_S=\{0.01,5,10,15,25,50,60,85,95,100\}\\[2mm]
T_E=\{0.01,5,10,15,25,50,60,85,95,100\}\\[2mm]
T_{ES}=\{0.01,5,10,15,25,50,60,85,95,100\}\\[2mm]
T_P=\{0.01,5,10,15,25,50,60,85,95,100\}
\end{array}$$

$$\mathcal{I}_C: S\in \langle 25,50\rangle,
E\in \langle 95,100\rangle,
ES\in \langle 0.01,5\rangle,
P\in \langle 0.01,10\rangle
$$

Projection of the approximated phase space to the enzyme/substrate plane is shown in Figure~\ref{fig:enzyme}. 
For both the oscillatory model and the enzyme kinetics model full version of Algorithm~\ref{alg:getsuccs} (with backward simulations) was used.

\begin{figure} 
\begin{center}
\vspace{-5mm}
\includegraphics[scale=.5]{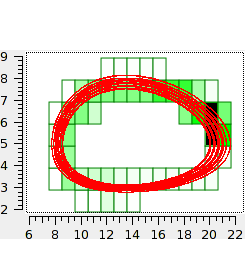}
\hspace{0cm}\includegraphics[scale=.5]{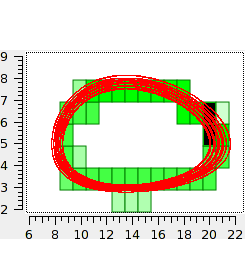}
\hspace{0cm}\includegraphics[scale=.4]{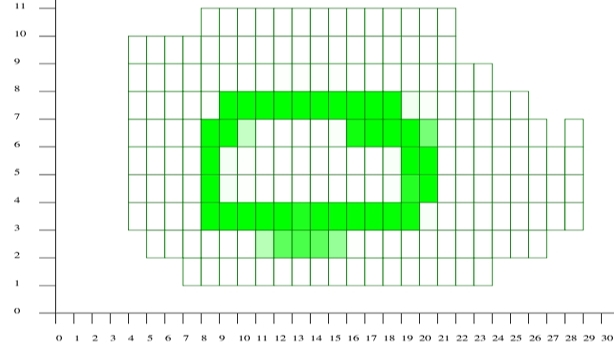}
\\
$\kappa=4$\hspace{2.5cm}$\kappa=16$\hspace{5cm}$\kappa=60$
\vspace{-3mm}
\end{center}
\caption{Reachability in oscillatory model and comparison with numerical simulation, first two figures were obtained using the full version of Algorithm~\ref{alg:getsuccs}, the third one with lines~\ref{line:back_begin}--\ref{line:back_end} omitted.
For comparison: using the rectangular abstraction transition system on this biochemical model, the whole phase space $[0,30]\times[0,12]$ is reachable from the same inital conditions.
}
\vspace{-3mm}
\label{fig:bayramov}
\end{figure}




\begin{table}
\begin{center}
\begin{tabular}{|c||c|c|c|c|c|c||c|c|c|c|}
\hline
& \multicolumn{6}{c||}{Oscillatory} & \multicolumn{4}{c|}{Enzyme}\\\hline
$\kappa$ &
4 & 8 & 16 & 32 & 64 & 128 &
4 & 5 & 6 & 7\\\hline
$|\mathcal{R}(\mathcal{I}_C)|$
& $52$ & $46$ & $40$ & $39$ & $37$ & $35$
& $76$ & $104$ & $123$ & $166$\\
$\varrho$ & 
$1.63$ & $2.2$ & $3.78$ & $2.9$ & $4.57$ & $6$ &
$4.36$ & $10.76$ & $16.8$ & $53.6$\\\hline
\end{tabular}
\end{center}
\caption{Results for the two models and several different settings of the discretization parameter~$\kappa$. 
}
\label{tab:results}
\end{table}

\begin{figure}
\begin{center}
\includegraphics[scale=.3]{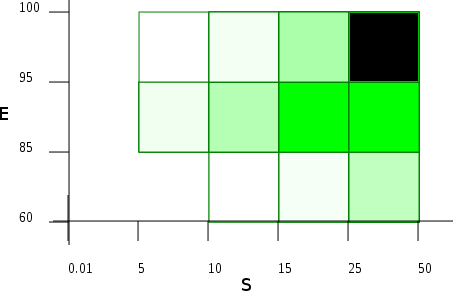}
~\includegraphics[scale=.3]{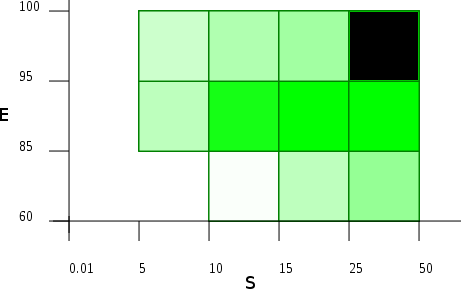}
~\raisebox{5mm}{\includegraphics[scale=.3]{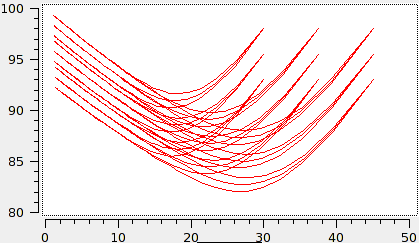}}\\
\hspace*{1.5cm}$\kappa=4$\hspace{3.3cm}$\kappa=6$\hspace{2cm}Numerical simulation
\vspace{-3mm}
\end{center}
\caption{Enzyme kinetics model -- projection of the reachable set to the enzyme/substrate plane and comparison with numerical simulation.}
\label{fig:enzyme}
\end{figure}


\subsection{Case Study on E.Coli Ammonium Assimilation Model}
\label{subsec:ammonium}

We consider a model specifying the ammonium
transport from the external environment into cells of \emph{E.~Coli}~\cite{ECMOANMODEL}. 
The model describes the ammonium transport process that takes effect
at very low external ammonium concentrations. In such conditions, the
transport process complements the deficient ammonium diffusion. The
process is driven by a membrane-located ammonium transport protein $AmtB$ that
binds external ammonium cations $NH_4ex$ and uses their electrical
potential to conduct $NH_3$ into the cytoplasm. In
Figure~\ref{fig:ammoniummodel}, biochemical reactions of this model and the scheme of the transport channel are
shown (left and middle).

 The level of pH and external ammonium concentration are considered constant. 

The system of differential equations:

$$\begin{array}{c@{$\,=\,$}l}
\frac{d[AmtB]}{dt} & -k_1 [AmtB] [NH_4ex] + k_2 [AmtB:NH_4] + k_4 [AmtB:NH_3]\\[1mm]
\frac{d[AmtB:NH_3]}{dt} & k_3 [AmtB:NH_4] - k_4 [AmtB:NH_3]\\[1mm]
\frac{d[AmtB:NH_4]}{dt} & k_1 [AmtB] [NH_4ex] - k_2 [AmtB:NH_4] - k_3 [AmtB:NH_4]\\[1mm]
\frac{d[NH_3in]}{dt} & k_4 [AmtB:NH_3] - k_6 [NH_3in][H_{in}] + k_7 [NH_4in] + k_9 [NH_3ex]\\[1mm]
\frac{d[NH_4in]}{dt} & k_{6} [NH_3in] [H_{in}] - k_5 [NH_4in] - k_7 [NH_4in]
\end{array}
$$

Constant species: $NH_3ex, NH_4ex, H_{in}, H_{ex}.$

Initial conditions and threshold numbers: 
$$
\hspace*{-1cm}\begin{array}{ccl}
T_{NH_3ex} & = &\{0, 28\cdot 10^{-9}, 29\cdot 10^{-9}, 1\cdot 10^{-5}\}\\
T_{NH_4ex}& = &\{0, 49\cdot 10^{-7}, 5\cdot 10^{-6}, 1\cdot 10^{-5}\}\\
T_{AmtB}& = &\{0,1\cdot 10^{-12},1\cdot 10^{-10},5\cdot 10^{-6},9.9\cdot 10^{-6},1\cdot 10^{-5} \}\\
T_{AmtB:NH_3}& = & \{0, 1\cdot 10^{-7}, 1\cdot 10^{-5}\}\\
T_{AmtB:NH_4}& = &\{0, 1\cdot 10^{-7}, 1\cdot 10^{-5}\}\\
T_{NH_3in}& = &\{0,1\cdot 10^{-8},1\cdot 10^{-7},1\cdot 10^{-6},11\cdot 10^{-7},1\cdot 10^{-5},1\cdot 10^{-4},1\cdot 10^{-3}\}\\
T_{NH_4in}& = &\{0,1\cdot 10^{-8},1\cdot 10^{-7},2\cdot 10^{-6},2.1\cdot 10^{-6},1\cdot 10^{-6},1\cdot 10^{-5},1\cdot 10^{-4},1\cdot 10^{-3}\}\\
\end{array}
$$

$$
\begin{array}{cl}
\mathcal{I}_C: & NH_3ex \in \langle 28\cdot 10^{-9}, 29\cdot 10^{-9}\rangle,\\
& NH_4ex \in \langle 49\cdot 10^{-7}, 5\cdot 10^{-6} \rangle, \\
& AmtB \in \langle 0, 1\cdot 10^{-5} \rangle,  \\
& AmtB:NH_3 \in \langle 0,1\cdot 10^{-5} \rangle,\\ 
& AmtB:NH_4 \in \langle 0,1\cdot 10^{-5} \rangle, \\
& NH_3in \in \langle 1\cdot 10^{-6},11\cdot 10^{-7}\rangle,\\ 
& NH_4in \in \langle 2\cdot 10^{-6},21\cdot 10^{-7} \rangle.\\
\end{array}
$$ 

\begin{figure*}
\hspace*{0cm}\scalebox{.8}{
\raisebox{1.8cm}{\scalebox{0.8}{\parbox{10cm}{
$$\begin{array}{cl}
AmtB + NH_4ex  \stackrel{k_1}{\leftarrow}\stackrel{k_2}{\rightarrow} AmtB:NH_4 & k_1=5\cdot 10^8, k_2=5\cdot 10^3\\
AmtB:NH_4 \stackrel{k_3}{\rightarrow} AmtB:NH_3 + H_{ex} & k_3 = 50\\
AmtB:NH_3 \stackrel{k_4}{\rightarrow} AmtB + NH_3in & k_4 = 50\\
NH_4in \stackrel{k_5}{\rightarrow} & k_5 = 80\\
NH_3in + H_{in} \stackrel{k_6}{\leftarrow}\stackrel{k_7}{\rightarrow} NH_4in & k_6=1\cdot 10^{15}, k_7=5.62\cdot 10^{5}\\
NH_3in \stackrel{k_8}{\leftarrow}\stackrel{k_9}{\rightarrow} NH_3ex & k_8=k_9=1.4\cdot 10^{4}
\end{array}
$$
}}}
\hspace{0.2cm}
\includegraphics[height=4.4cm]{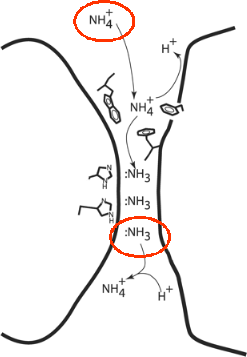}
\hspace{0.2cm}
\includegraphics[height=4.4cm]{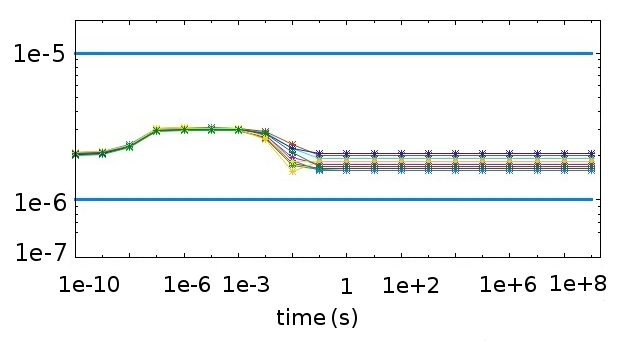}
\hspace*{-5mm}
}
\vspace{-0.5cm}
\caption{Ammonium transport model (left). 
Simulations of the ammonium assimilation model from $20$ randomly sampled points in $\mathcal{I}_C$ projected on the concentration of $NH_4in$, blue lines represent bounds on this concentration found by the QDAA - two subsequent thresholds $10^{-6},10^{-5}$ (right).}
\label{fig:ammoniummodel}
\end{figure*}

The upper bounds on concentrations of $NH_3in$ and $NH_4in$ considering the biological system with given initial conditions were estimated as $1.1 \cdot 10^{-6}$ ($NH_3in$ does not exceed the initial concentration) and $5.4 \cdot 10^{-4}$ by the rectangular abstraction (overapproximation).

Reachable intervals using Algorithm~\ref{alg:getsuccs} without the backward simulations 
were
$[10^{-8},1.1\cdot 10^{-6}]$ for $NH_3in$ ($NH_3in$ does not exceed the initial concentration),
and 
$[ 10^{-6},10^{-5}]$ for $NH_4in$. 
This results are in agreement with simulated data and in the case of the concentration of $NH_4in$ the QDAA results are by one order closer to numerical simulations than the rectangular abstraction results as can be seen in the right part of Figure~\ref{fig:ammoniummodel}. 




\section{Conclusion}
\label{sec_conclusion}

We have presented a new theoretical method for finite discrete approximation of autonomous continuous systems equipped with a measure that indirectly quantifies correspondence of the approximated behaviour with the original continuous behaviour. 
We have provided a computational technique which we implemented in a prototype software. We have examined the implementation on small dimensional models which showed satisfactory results for computing reachability.

The method can be either used as a parameterized simulation technique or employed with rectangular abstraction to quantify the extent of spurious counterexamples. Thus the method can improve the current possibilities of analysis based on model checking techniques. We leave for future work integration of this method into the software for model checking of biochemical dynamical systems~\cite{BBS10}.

At the theoretical side, we leave for future work precise clarification of our method wrt the rectangular abstraction. From the computational viewpoint, we aim to develop a parallel reachability algorithm that would make the method scalable and applicable to systems of larger dimensions.


\bibliographystyle{eptcs} 



\end{document}